\documentclass[acmsmall]{acmart}

\def\fullFlag{0}

\AtBeginDocument{%
	}

\usepackage{amsmath,amssymb}
\usepackage{bm}
\usepackage{graphicx}
\usepackage{subcaption}
\usepackage{hyperref}
\usepackage{mathtools}
\usepackage{array}
\usepackage{dashbox}
\usepackage[all,cmtip]{xy}
\usepackage[table,dvipsnames]{xcolor}
\usepackage{prftree}

\usepackage{tikz}
\usetikzlibrary{trees}
\usepackage{xspace}

\usepackage{algorithm}
\usepackage{algpseudocode}
\usepackage[shortlabels]{enumitem}
\newcommand{\ie}{\textit{i.e.}\xspace}
\newcommand{\eg}{\textit{e.g.}\xspace}
\newcommand{\etc}{\textit{etc.}\xspace}

\newtheorem{definition}{Definition}
\newtheorem{theorem}{Theorem}
\newtheorem{lemma}{Lemma}

\newtheorem{corollary}{Corollary}

\newcommand{\mycircled}[2][none]{
 \tikz[baseline = (a.base)]\node[draw, circle, inner sep = 1pt, outer sep = 0pt, fill = #1](a){\ensuremath #2\strut};
}

\newcommand{\s}{\mathsf}

\newcommand{\inp}{\scalebox{0.85}{${\boxdot}$}}
\newcommand{\oup}{\scalebox{0.85}{${\boxtimes}$}}
\newcommand{\inm}{\mathsf{im}}
\newcommand{\oum}{\mathsf{om}}

\newcommand{\tx}{\mathsf{tx}}
\newcommand{\id}{\mathsf{id}}

\newcommand{\ith}{i^{\mathrm{th}}}
\newcommand{\jth}{j^{\mathrm{th}}}
\newcommand{\bmtop}{\boldsymbol{\top}}
\newcommand{\bmbot}{\boldsymbol{\bot}}

\newcommand{\ri}{\mathsf{ri}} 

\newcommand{\concat}{\mathbin{\|}}

\newcommand{\aeval}{\mathsf{aEval}_{\Gamma, \tx}}
\newcommand{\geval}{\mathsf{gEval}_{\Gamma, \tx}}
\newcommand{\beval}{\mathsf{bEval}_{\Gamma, \tx}}
\newcommand{\nonf}{\{ \mathsf{U}, \bmtop \}}
\newcommand{\ctvp}{CTV\raisebox{0.5ex}{\Large +}\xspace}

\newcommand{\accept}{\mathsf{Accept}}

\algnewcommand{\IfReturn}[2]{%
	\State \algorithmicif\ #1 \algorithmicthen\ \algorithmicreturn\ #2%
}
\algnewcommand{\ForDo}[2]{%
	\State \algorithmicfor\ #1 \algorithmicdo\ #2%
}

\newcommand{\acmdashbox}[2][]{%
\tikz[baseline=(B.base)]{%
\node[draw,dashed,inner sep=2pt,outer sep=0pt,rounded corners=0pt,#1] (B) {#2};%
}%
}

\newcommand{\mysmallskip}{\par\everypar{}\setlength{\parskip}{3pt}\noindent}

\newcommand\itpara[1]{\mysmallskip\noindent\textit{#1}}

\newcommand{\CaptionAndDescription}[1]{%
	\caption{#1}%
	\Description{#1}%
}

\setcopyright{cc}
\setcctype{by}
\acmDOI{10.1145/3839489}
\acmYear{2026}
\acmJournal{PACMPL}
\acmVolume{10}
\acmNumber{OOPSLA2}
\acmArticle{357}
\acmMonth{10}
\acmSubmissionID{oopslab26main-p815-p}
\received{2026-03-17}
\received[accepted]{2026-06-10}

\begin{document}
\hyphenation{Sherman}

\author{Shuyang Tang}
\affiliation{
\institution{Shanghai\hspace{0.11em} Jiao\hspace{0.11em} Tong\hspace{0.11em} University}
	\city{Shanghai}
	\country{China}}
\email{htftsy@gmail.com}
\orcid{0000-0001-9379-3609}

\author{Sherman S. M. Chow}
\affiliation{
	\institution{The Chinese University of Hong Kong}
	\city{Shatin}
	\state{N.T.}
	\country{Hong Kong}}
\email{smchow@ie.cuhk.edu.hk}
\orcid{0000-0001-7306-453X}

\author{Hongfei Fu}
\authornote{Corresponding author.}
\affiliation{%
	\institution{Shanghai University of Finance and Economics}
	\city{Shanghai}
	\country{China}
}
\email{fuhongfei@mail.sufe.edu.cn}
\orcid{0000-0002-7947-3446}

\author{Zihan Guo}
\affiliation{
	\institution{Sun Yat-sen University}
	\city{Guangzhou}
	\country{China}}
\email{guozh29@mail2.sysu.edu.cn}
\orcid{0000-0002-2401-9174}

\author{Guoqiang Li}
\affiliation{
	\institution{Shanghai Jiao Tong University}
	\city{Shanghai}
	\country{China}}
\email{li.g@sjtu.edu.cn}
\orcid{0000-0001-9005-7112}


\renewcommand{\shortauthors}{Tang, Chow, Fu, Guo, and Li}

\title{Staged Multi-step UTXO Workflows via Recursive Invariants}

\begin{abstract}
Stateless UTXO-style execution validates transactions using local and referenced data, enabling parallel validation and predictable serialized-size/weight accounting. 
However, multi-step workflows must thread state across outputs, and a prepared next-step transaction may become stale when another valid spend confirms first. 
Explicit state threading therefore shifts consistency maintenance, off-chain tracking, and transaction rebuilding onto the protocol boundary, potentially increasing coordination cost and latency. 
Recursive invariants (RIs), our proposed transaction-level logic and toolchain, address this gap by expressing workflow rules as transaction-level predicates over a transaction's inputs and indexed successor positions referenced by the RI. 
Modeled this way, an accepted transaction that realizes such a successor position re-checks the predecessor's RI one step later, carrying the workflow rule forward without introducing application-level shared mutable state or executable logic attached to outputs. 
Accordingly, multi-step protocol rules preserve validation-time locality and admit explicit cost accounting, while cross-transaction guarantees arise from repeated one-step checking. 
Not all successor clauses are checkable when the current transaction is validated, so our small statically typed domain-specific language (DSL) uses three-valued semantics over true, false, unknown to defer future-dependent obligations until they become checkable. 
Co-designed with this DSL, our framework formalizes UTXO validation and ledger extension, identifies the validation-time-evaluable one-step fragment, and proves the deduction system sound with respect to the three-valued semantics. 
Here, we also give validation and ledger-extension algorithms corresponding to the formal model. 
On the systems side, we implement a prototype RI interpreter and benchmarking toolchain for the six reported workloads. 
With six practice-motivated case studies, the reported benchmark traces exhibit approximately linear cumulative validation-cost proxy growth, while illustrating staged workflow constraints without committing each step to a preconstructed successor transaction.
\end{abstract}



\maketitle

\section{Introduction}
Blockchains are append-only ledgers that record validated transactions.
Smart-contract applications, which span domains from financial services to the Internet of Things~\cite{nss/MengWWLYLZC18}, commonly follow one of two execution models:
account-based systems and \emph{unspent transaction output}
(UTXO)-based systems.
Account-based systems, \eg, \emph{Ethereum}, maintain a shared global state that contracts read and write.
UTXO-based systems, \eg, \emph{Bitcoin}, represent state as UTXOs.

Account-based chains offer shared storage and general programmability, but shared state introduces serial effects, reentrancy, and state-dependent execution work that complicate ex-ante resource accounting.
Fees remain market-dependent in either model;
we focus on local validation dependencies and explicit resource accounting, for which UTXO execution is attractive.
In a UTXO-based system, a transaction consumes prior outputs as \emph{inputs} and creates new \emph{outputs}.
Scripts specify spending conditions but do not provide persistent contract storage~\cite{fcw/ChakravartyCMMJ20}.

\subsection{Problem and Approach}
\itpara{Multi-step UTXO Workflows are Hard.}
Bitcoin-like UTXO validation is largely local to a transaction and its referenced UTXOs, enabling substantial parallel checking and predictable serialized size of the transaction encoding and Bitcoin-style transaction weight~\cite{fgcs/BartolettiMZ26,fc/CromanDEGJKMSSS16}.
Because validation inspects only data carried or referenced by a transaction, the resulting proof obligations are local and the validation work can be accounted for explicitly.
The cost is that multi-step workflows must carry their state explicitly through outputs.
Each step spends a UTXO encoding protocol state and creates a successor UTXO with updated state.
Before confirmation, transaction ordering is provisional, so a prepared successor transaction can become stale if another valid spend confirms first.
This occurs due to mempool competition, fee-driven reordering, delayed confirmation, and protocol branches such as timeout, refund, and dispute paths.
Participants must therefore track the latest state-carrying UTXO off-chain and rebuild after conflicts, increasing coordination and latency.

Bitcoin-like UTXO scripts support per-transaction checks on signatures, timelocks, and covenant-style templates such as CTV-op
(\textsf{OP\_CHECKTEMPLATEVERIFY})~\cite{fc/OConnorP17,fc/MoserES16,bip119}.
These mechanisms strengthen transaction constraints, but they do not provide persistent shared mutable contract storage or a reusable interface for multi-step workflow invariants.
A multi-step workflow is a sequence of transactions that jointly implement a protocol, while shared mutable contract storage allows application state to be updated and reused across steps.
Some higher-expressivity designs rely on optimistic disputes, which make worst-case on-chain cost adversary-dependent;
richer covenant designs may still require explicit state passing and off-chain coordination.

We instead enforce multi-step rules with validation-time locality and future-dependent checking bounded to one successor step, with explicit accounting of validation work.
We call this transaction/validation model and RI interface \ctvp.
Rather than commit to a specific successor transaction, an RI constrains whichever transaction later realizes a referenced successor position.

\itpara{Recursive Invariants.}
We propose \emph{recursive invariants} (RIs), transaction-level predicates over inputs ($\inp_i$) and indexed successor positions ($\oup_i$);
a position referenced by a successor operator is \emph{designated}, with no separate designation field.
An accepted realization of such a position re-checks the inherited RI one step later, carrying workflow rules forward without application-level shared mutable state or executable output logic.
Viewed against account-based execution, RIs separate persistent workflow rules from shared mutable contract state: staged protocols may need the former without the latter.

\itpara{Validation-Time Locality and Future-Dependent Clauses.}
The validator checks the current transaction and inherited RIs of its parents;
successive such checks thereby enforce cross-transaction guarantees without exploring future paths.
Some successor-dependent clauses cannot be decided when the current transaction is validated, so our statically typed DSL uses Kleene's strong three-valued semantics: such a clause is $\mathsf{unknown}$ until realized and then re-checked.
The validation-time locality claims apply to the evaluable one-step fragment.
More generally, deferred obligations and the evaluable-fragment restrictions keep RI evaluations stable at unrelated prior transactions, yielding parent-local rather than ledger-wide re-checking.

\itpara{Formal Results.}
We formalize UTXO validation and ledger extension and give a judgmental deduction system sound for the three-valued semantics.
These results apply to the evaluable fragment.

\itpara{Turing-completeness.}
Under mathematical-integer semantics and an unbounded transaction chain, the evaluable one-step fragment simulates UTXO counter machines and is Turing-complete~\cite{csf/BartolettiLZ21};
\S\ref{subsec:turing} gives the construction.

\itpara{Implementation and Evaluation Preview.}
We implement an RI benchmark interpreter and cost-estimation toolchain and evaluate six workflows~(\S\ref{sec:evaluation}).
\ifnum\fullFlag = 1
Appendix~\ref{app:future}
\else
Appendix~C of our full version~\cite{eprint/TangCFGL26}
\fi
discusses additional limitations and future directions.

Our contributions are:
(1) RIs and a statically typed, three-valued DSL over inputs and indexed successors;
(2) a formal UTXO execution model, one-step evaluability conditions, a sound deduction system, and a counter-machine simulation~\cite{csf/BartolettiLZ21}; and 
(3) an RI benchmark interpreter and cost-estimation toolchain whose six workloads show approximately linear cumulative validation-cost proxy growth over the evaluated traces.

The BIP-119 opcode \textsf{OP\_CHECKTEMPLATEVERIFY} is often called CTV~\cite{bip119}.
We call it CTV-op here to distinguish it from \ctvp.

\subsection{Scope and Limitations}
\label{subsec:limitations}
\ctvp and neighborhood covenants, which attach constraints per output, are both Turing-complete;
our distinction is validation interface and validation-time locality, not worst-case computability.
\ctvp checks the current transaction, referenced predecessors, and RI-referenced successor positions, but has no native ledger-wide scan or quantification over unrelated UTXOs;
such properties require carried summaries, commitments, or authenticated witnesses.
It is not a general-purpose shared-mutable-state platform.

Validation-time locality concerns ledger dependencies:
checking follows transaction relationships selected by the transaction and RI syntax rather than scanning unrelated ledger state.
It does not imply constant time for arbitrary RI size, nesting depth, or operand cost.
Parent-local re-checking concerns which previously accepted RI roots must be re-evaluated after an extension, not the complete read set of an RI evaluation;
each checked RI may follow the syntax-delimited transaction dependencies permitted by the language.
The formal semantics uses mathematical domains, whereas the benchmark implementation uses finite $256$-bit representations.
RI safety does not imply liveness or recovery:
successor clauses constrain realizations if they occur, and an evaluable RI formula may admit no continuation;
retries, timeouts, refunds, and satisfiability checks require additional protocol logic or tooling.

The concrete checker stores at most one RI formula and one checker-side instance label per transaction;
the denotational semantics abstracts from labels.
Continuation-marker values/tags encode modes;
supporting multiple named formulas or instances would require a richer representation.
\subsection{Related Work}
\label{sec:related}

\itpara{UTXO Systems and Covenants.}
Neighborhood covenants~\cite{csf/BartolettiLZ21} are our closest baseline for UTXO programmability.
Given their comparable worst-case expressiveness, we compare their validation interfaces and cost distribution along three main axes.

First, neighborhood covenants attach constraints per output;
\ctvp states one transaction-level invariant over its inputs and referenced successor positions.
This avoids duplicating a shared transaction-wide rule;
output-attached constraints remain natural for per-output state.

Second, \ctvp supports nested modality, allowing one invariant to dereference a syntax-bounded number of earlier steps.
Neighborhood-covenant encodings can express similar dependencies by propagating selected fields through intermediate transactions or outputs.
The tradeoff is cost placement: nesting pays for modal dereferencing, while cloning pays for repeated state propagation.
\ifnum\fullFlag = 1
Appendix~\ref{app:usecase} gives a small rollback micro-example.
\else
Appendix~B.1 of the full version~\cite{eprint/TangCFGL26} gives a small rollback micro-example.
\fi

Third, \ctvp interprets future-dependent RI clauses using Kleene's three-valued semantics, distinguishing facts fixed at validation from future-dependent facts.
This makes one-step successor obligations explicit in the validation semantics and accounts for their validation work when they become checkable.
Neighborhood-covenant encodings can likewise enforce future constraints through an output-attached interface.

CTV-op and related covenant proposals constrain future spends through template commitments and support constructions such as congestion control, vaults, and payment pools~\cite{bip119,BIP345_OPVAULT,Optech_Topic_Vaults,Optech_Topic_CTV}.
CTV-op itself constrains an output's next spend, while \ctvp provides one transaction-level RI over inputs and indexed successor positions for multi-step constraints.
Extended UTXO (EUTXO) supports forward chaining and branching through data-carrying outputs and script contexts~\cite{fcw/ChakravartyCMMJ20,lsfa/VinogradovaS24}.
Both EUTXO and \ctvp support branching, but place continuation rules differently: EUTXO organizes continuation state per output, while \ctvp states one transaction-level RI over several successor positions.
Both neighborhood covenants and \ctvp can encode timelocks and incorporate externally determined values once the relevant transaction data are fixed.
We use timelocks and block-derived values in the case studies as stress tests for validation-time locality and explicit accounting of validation~work.

\itpara{Contract Languages and Verification Frameworks.}
BitML gives a calculus and verification framework for Bitcoin contracts, with emphasis on payment values and contract timing~\cite{ccs/BartolettiZ18}.
BitVM broadens Bitcoin expressivity by moving complex computation off-chain and enabling on-chain verification via fraud proofs~\cite{bitvm}.
Scilla~\cite{pacmpl/SergeyNJ0TH19} is an intermediate-level language designed for safer smart contracts, with automata-based semantics that support verification and formal reasoning.
Solidity~\cite{esop/HajduJ20} and Move~\cite{tool/move} target account-based execution models with shared mutable state, whereas Plutus~\cite{tool/plutus} targets Cardano's EUTXO model.
These systems provide programming, state-management, or off-chain-computation abstractions;
\ctvp instead fixes a transaction-level validation discipline for staged UTXO workflows.

\itpara{Ledger Semantics.}
Prior work on program logics for ledgers~\cite{fmbc/MelkonianSC25} models ledgers as sequences of state-transforming transactions and develops a Hoare-style concurrent separation logic for modular reasoning about ledger fragments.
\ctvp operates at the UTXO validation layer with recursive invariants, three-valued semantics for future-dependent clauses, and evaluability conditions supporting one-step successor obligations and parent-local re-checking.


\section{Operational Overview}
\label{sec:overview}
We first give the RI viewpoint, a running example, and a recursion-only example.
We use $[n] = \{1, \ldots, n\}$;
$|v|$ for sequence/set size (or integer absolute value);
$\emptyset/\epsilon$ for empty set/sequence;
and $\Omega/\s{U}$ for unknown field/Boolean values.

\itpara{Three-Valued Semantics for Deferred Obligations.}
Validation may precede realization of an RI-referenced successor.
Such future-dependent clauses use Kleene's K3 semantics~\cite{book/Kleene52}: $\s{U}$ records a deferred obligation until the successor appears.
Simply substituting $\bmtop$ would be non-compositional under negation (\eg, $\neg\oup_1 p$ would spuriously fail before realization).

\itpara{Main Idea.}
Each transaction carries one RI rather than per-output scripts.
$\inp_i\varphi$ evaluates $\varphi$ on input~$i$, while $\oup_i\varphi$ evaluates it on the transaction realizing successor position~$i$;
unmodalized fields belong to the current transaction.
Thus $\bigwedge_{j \in [\ell]}\oup_j p_j$ expresses per-position restrictions.
A successor position is realized only when a later transaction references it;
the realizing transaction and realized amount are then determined from the ledger.

\itpara{An Illustrative Example.}
Consider coins held by a junior, with a spending cap on each transaction.
A transaction $\tx_0$ can collect the coins as inputs and carry the recursive invariant
$
	F \coloneqq \text{\acmdashbox{$\oup_3 \bmbot\ \wedge\ \oup_1 \s{RI} \ \wedge\ \s{pk} = \oup_1 \s{pk} \ \wedge\ \oum_2 < 100$}}.
$

In the intended continuation, positions~1 and~2 carry retained and spent amounts $\oum_1, \oum_2$.
$\oup_3 \bmbot$ forbids position~3 and, by gap-free indexing, all later positions.
$\oup_1\s{RI}$ requires the realizing transaction to carry the same RI formula, thereby propagating it through position~1;
$\s{pk} = \oup_1 \s{pk}$ preserves its key;
and $\oum_2<100$ bounds position~2 if realized.

Before position~1 of $\tx_0$ is realized, $\oup_1\s{pk}$ evaluates to $\Omega$ and its Boolean clause to $\s{U}$.
When $\tx_1$ realizes that position, the RI carried by $\tx_0$ is re-checked and must not evaluate to $\bmbot$.
Re-checking reuses $\tx_0$'s committed inclusion context, so $\Lambda$ and $\s{Rnd}$ retain their values from its original inclusion.

For example, if $\tx_1$ changes the public key, the previously unknown clause $\s{pk} = \oup_1 \s{pk}$ becomes false when $\tx_1$ realizes position~1, so re-checking rejects it.

\itpara{Operational Reading.}
Inputs are ordered and realized successor positions are gap-free;
$\oup_i\s{RI}$ propagates the RI formula, and successor-dependent clauses are re-checked when realized.
The validator thus checks the candidate and inherited parent RIs.
To ensure parent-local re-checking, ordinary modal operands must be \emph{pure}: they contain no successor modality or ledger-derived $\s{out}_j/\oum_j$, so later realization of an ancestor output cannot change an accepted descendant (\S\ref{subsec:dsl}).

\itpara{A Compact Recursion Example.}
Besides constraining successor realizations, an RI can also constrain the current transaction and propagate to a successor.
Consider the invariant
$F \coloneqq (\oup_1 x_1 = x_1 + \inp_1 x_1) \wedge \oup_1 \s{RI}$.
\acmdashbox{$\oup_1 x_1 = x_1 + \inp_1 x_1$} is a one-step constraint across three consecutive positions in the chain.
Let $\tx_0,\tx_1,\ldots$ denote a continuation chain in which $\tx_{i + 1}$ realizes successor position~1 of $\tx_i$.
For $\tx_i$ in this chain, the unmodalized $x_1$ refers to $\tx_i.x_1$, 
$\inp_1x_1$ to $\tx_{i - 1}.x_1$, and
$\oup_1x_1$ to $\tx_{i + 1}.x_1$ when that successor is realized.
Thus, once $\tx_{i + 1}$ is proposed, the RI enforces
$\tx_{i + 1}.x_1 = \tx_i.x_1 + \tx_{i - 1}.x_1$.
\acmdashbox{$\oup_1 \s{RI}$} requires the transaction realizing successor position~1 to carry the same RI formula, so the constraint is recursively enforced along that continuation chain.
The base case is supplied externally by fixing the first two transactions in the chain, \eg, $\tx_0.x_1 = \tx_1.x_1 = 1$.
The RI then constrains all subsequent transactions in the chain, yielding the Fibonacci progression.
\section{Execution Model and Recursive-Invariant Language}
This section defines our execution model and RI interface.
\S\ref{subsec:ledgermodel} formalizes ledger extension and the transaction acceptance predicate $\accept(\Gamma, \tx)$, while \S\ref{subsec:dsl} defines the RI DSL and its \emph{evaluable fragment}.
\S\ref{subsec:algo} gives concrete transaction-verification and ledger-extension algorithms invoking \textsc{CheckRI}, the implementation-level counterpart of the denotational evaluator $\beval$ (\S\ref{subsec:semantics}).
\S\ref{subsec:turing} establishes Turing-completeness by simulating UTXO counter machines.
The semantics and soundness results in \S\ref{sec:dslri} apply to the evaluable fragment.


\subsection{Block Execution Model}
\label{subsec:ledgermodel}

We model the ledger as a canonical, append-only sequence of validated transactions, abstracting away consensus and reorgs because our results concern validation-time rules and ledger extension.
We fix an initial ledger $\Gamma_0$ containing the genesis state (\eg, coinbase-style value creation and pre-existing UTXOs).
Execution is modeled by a ledger-extension function $\s{Ext}(\Gamma, \tx)$ and a transaction acceptance predicate $\accept(\Gamma, \tx)$.
Inputs reference earlier transactions and identify output positions realized when the spending transaction is accepted.
Acceptance enforces parent existence, fresh gap-free output positions, non-negative amounts, balance consistency, and RI admissibility.
We now define the transaction structure and acceptance conditions.
Let $\mathbb{G}$ be the ambient group used for public keys, signatures, and group witnesses.

\begin{definition}[Transaction]
A transaction in our model is typed as
\[
	\tx \in \mathcal{M}_{\tx}
	=
	\mathbb{Z}
	\times \mathbb{G}
	\times (\mathbb{G})^+
	\times (\mathbb{Z})^+
	\times (\mathbb{G})^+
	\times ((\mathbb{Z})^+)^3
	\times bexpr
	\times (\mathbb{Z})^\ast
\]
and comprises:
\begin{enumerate}[nosep, topsep = 0pt]
\item
\emph{A unique identifier \underline{$(\tx.\id \in \mathbb{Z})$}.}

\item
\emph{A proposer public key \underline{$(\tx.\s{pk} \in \mathbb{G})$}.}
We assume $\tx.\s{pk} = \tx.y_1$.

\item
\emph{A tuple of signatures \underline{$(\tx.\sigma \in (\mathbb{G})^+)$}.}

\item
\emph{A tuple of integers \underline{$(\tx.x \in (\mathbb{Z})^+)$}.}
Elements in $\tx.x$ are denoted $x_1, x_2$, \etc

\item
\emph{A tuple of group members \underline{$(\tx.y \in (\mathbb{G})^+)$}.}
The transaction proposer's public key corresponds to the first element $y_1$.

\item
\emph{A tuple of inputs
\underline{$(\tx.\s{in} \in (\mathbb{Z})^+)$}
with associated amounts
\underline{$(\tx.\inm \in (\mathbb{Z})^+)$}.}
Each $\s{in}_i$ references the identifier of the $\ith$ input transaction.
The amount $\inm_i$ is the value assigned to the referenced output position by this input.

\item
\emph{A tuple of output indexers
\underline{$(\tx.\s{oi} \in (\mathbb{Z})^+)$}.}
For each input $\s{in}_i$, $\s{oi}_i = j$ means that this input realizes and spends the $\jth$ output position of the referenced transaction.
Output-spender identifiers and output amounts are not immutable components of $\tx$.
Instead, they are ledger-derived metadata determined by later spending transactions, as defined below.

\item
\emph{A recursive invariant ($\tx.\ri \in bexpr$) with a possibly empty sequence of continuation markers ($\tx.\lambda \in (\mathbb{Z})^{\ast}$).}
\end{enumerate}
\end{definition}

The model has no implicit per-output owner: $\s{Ver}(y_1, \sigma_1)$ authenticates the current transaction only, so any required successor ownership must be bound explicitly by the RI.

For a ledger $\Gamma$, we write $\s{lift}(\Gamma, \id)$ for the set of transaction identifiers appearing in $\Gamma$.
The notation extends analogously to other data fields when needed.

For denotational convenience, we write $\Gamma[\id']$ for the unique transaction $\tx' \in \Gamma$ such that $\tx'.\id = \id'$.
If no such transaction exists, we write $\Gamma[\id'] = \bot_{\tx}$.
This mapping is well-defined for a valid ledger because identifiers do not collide.

For $\tx \in \Gamma$, define the realized output-position set
\[
	\mathsf{OutPos}_{\Gamma}(\tx)
	\coloneqq
	\{j \in \mathbb{N}_{>0} \colon
		\exists\, \tx' \in \Gamma, k,
		\tx'.\s{in}_k = \tx.\id
		\land
		\tx'.\s{oi}_k = j\}.
\]
If a unique pair $(\tx',k)$ realizes position $j$, we define
$\tx.\s{out}^{\Gamma}_j \coloneqq \tx'.\id$
and
$\tx.\oum^{\Gamma}_j \coloneqq \tx'.\inm_k$.
Thus, $\s{out}_j$ and $\oum_j$ are \emph{ledger-derived output metadata}: semantically they are deterministic functions of the ledger, not immutable fields serialized in $\tx$.
An implementation may materialize these values in checker-maintained indexes or reconstruct them from ledger history.
If no unique realizing input exists, these quantities are undefined.
When the ledger is clear, we omit the superscript $\Gamma$.
The semantics in \S\ref{subsec:semantics} maps undefined ledger-derived fields to unknown.
Thus, an indexed position is a potential continuation slot: its spender/amount arise only upon realization, which successor clauses constrain but do not force.

A ledger $\Gamma \in \mathfrak{M}$ is a linear sequence of transactions, \ie,
$\mathfrak{M} = (\mathcal{M}_{\tx})^*$.
For a proposed transaction $\tx$, write
$\Gamma^+ \coloneqq \Gamma \cdot \tx$
for the candidate sequence obtained by appending $\tx$ to $\Gamma$.
All ledger-derived $\s{out}$ and $\oum$ fields in the acceptance test below are interpreted relative to the indicated ledger, in particular $\Gamma^+$ when checking the candidate extension.

For any $\Delta \in \mathfrak{M}$ and $q \in \Delta$, define the ledger-relative balance
$
	\mathsf{bal}_{\Delta}(q)
	\coloneqq
	\sum_{i = 1}^{|q.\s{in}|} q.\inm_i
	-
	\sum_{j \in \mathsf{OutPos}_{\Delta}(q)} q.\oum^{\Delta}_j
$.
%
We now define the condition \acmdashbox{$\accept(\Gamma, \tx)$} for accepting $\tx$.

Let
$P \coloneqq \{\tx.\s{in}_k \colon k \in [|\tx.\s{in}|]\}$
be the set of distinct parent identifiers referenced by $\tx$.

\begin{enumerate}[topsep = 0pt, label = {(o\arabic*)}]
\item \emph{Referenced output positions are well-formed, fresh, and gap-free.}
\begin{enumerate}[nosep, topsep = 0pt, label = {(\alph*)}]
\item \emph{Input tuples are aligned and referenced transactions exist:}
\[
	|\tx.\s{in}| = |\tx.\inm| = |\tx.\s{oi}|,
	\qquad
	\tx.\s{in}_1, \tx.\s{in}_2, \ldots \in \s{lift}(\Gamma, \id).
\]

\item
For each $p \in P$, let
$
	K_p \coloneqq \{k \in [|\tx.\s{in}|] \colon \tx.\s{in}_k = p\},
	J_p \coloneqq \{\tx.\s{oi}_k \colon k \in K_p\},
$
and let $m_p \coloneqq |\mathsf{OutPos}_{\Gamma}(\Gamma[p])|$.
We require
$
    |J_p| = |K_p|
$
and
$
    J_p = \{m_p + 1, \ldots, m_p + |K_p|\}
$.
Thus, no input duplicates an output position, no realized position is reused, and each parent is extended only by the next gap-free output position or positions.
\end{enumerate}

\item \emph{Non-negative amounts:}
for every input position $k$,
\acmdashbox{$\tx.\inm_k \ge 0$.}

\item \emph{Balance consistency:}
for the proposed transaction and every affected parent, \\
\acmdashbox{$\mathsf{bal}_{\Gamma^+}(q) \ge 0$ for every $q \in \{\tx\} \cup \{\Gamma[p] \colon p \in P\}$.}

\item \emph{Unique identifier:}
\acmdashbox{$\tx.\id \notin \s{lift}(\Gamma, \id)$.}

\item \emph{Recursive invariant admissibility:}
\acmdashbox{$(\Gamma^+, \tx) \models \tx.\ri$.}

\item \emph{Input verification:}
\label{cond:input_verification}
For every distinct parent identifier $p \in P$, we require
\acmdashbox{$(\Gamma^+, \Gamma[p]) \models \Gamma[p].\ri$.}
The admissibility relation $\models$ is defined in \S\ref{subsec:semantics}.
\end{enumerate}

If $\accept(\Gamma, \tx)$ holds, we define
$\s{Ext}(\Gamma, \tx) \coloneqq \Gamma^+$;
otherwise,
$\s{Ext}(\Gamma, \tx) \coloneqq \Gamma$.
No mutable transaction state beyond the ledger-derived metadata determined by the sequence is introduced.

Condition~\ref{cond:input_verification} is where deferred obligations become checkable.
A prior RI may evaluate a successor-dependent clause as $\s{U}$ when the prior transaction is first accepted.
Once the current transaction realizes the referenced successor position, the corresponding $\s{out}$ and $\oum$ metadata become defined in $\Gamma^+$, and the inherited clause is re-evaluated on the realized successor.
The candidate extension is accepted only if every affected parent's RI remains admissible.

Finally, we define valid ledgers as the background model for the DSL and its soundness statements.

\begin{definition}[Valid Ledger]
Fix an initial ledger $\Gamma_0$ satisfying the underlying issuance rules.
The valid ledgers are exactly those reachable from $\Gamma_0$ by successive accepted extensions:
$\Gamma_0$ is valid, and if $\Gamma$ is valid and $\accept(\Gamma, \tx)$ holds, then $\Gamma \cdot \tx$ is valid.
\end{definition}

By construction, every valid ledger has predecessor-before-successor order, collision-free identifiers, unique realization of each output position, gap-free realized output positions, non-negative realized amounts, and balance consistency.


\subsection{Recursive Invariants}
\label{subsec:dsl}
We formalize expressions, propositions, validation-time locality, and recursive invariants (RIs) as a domain-specific language for UTXO-like transaction scripts.
We first define arithmetic expressions $aexpr$ and group expressions $gexpr$.
For hash arguments, let $e$ range over either arithmetic or group expressions.
An arithmetic expression $a \in aexpr$ is defined by the grammar:
\[
\begin{aligned}
a \coloneqq {} & z \mid -a \mid a_1 + a_2 \mid a_1 \cdot a_2
\mid \s{H}(e_1, \ldots, e_\ell) \mid \lvert a \rvert \mid {}\\
& \inp_i a \mid \oup_i a \mid \Omega \mid \Lambda \mid \s{Rnd} \mid s_a,
\end{aligned}
\]
where $z \in \mathbb{Z}$, $s_a \in S_a$, and $\ell \geq 1$.
We use distinct typed hashes $H_a\colon\{0,1\}^*\to\mathbb{Z}$ and $H_g\colon\{0,1\}^*\to\mathbb{G}$, both written $\s{H}$ and applied to a canonical, type-tagged serialization;
context determines the output sort.
For readability, when hashing a group expression to an integer, we write
$\s{hashToInt}(g) \coloneqq \s{H}(g)$,
where $\s{H}(g)$ is interpreted in an arithmetic-expression context.

The operator $\inp_i$ evaluates an expression in the $\ith$ input transaction, while $\oup_i$ evaluates it in the transaction that realizes successor position~$i$.
The symbol $\Omega$ denotes an unknown evaluation at validation time.
The symbols $\Lambda$ and $\s{Rnd}$ denote the block height and an exogenous random value supplied by the transaction's canonical inclusion context;
that context is fixed upon inclusion and reused on later re-checking.
The symbol set $S_a$ is:
\[
S_a \coloneqq \{ \id, x_1, x_2, \ldots, \s{in}_1, \s{in}_2, \ldots, \s{out}_1, \s{out}_2, \ldots, \inm_1, \inm_2, \ldots, \oum_1, \oum_2, \ldots, \s{oi}_1, \s{oi}_2, \ldots \}.
\]
The sans-serif symbols (\eg, $\s{in}_i$, $\inm_i$, $\s{oi}_i$) are terminals distinct from the modal operators $\inp_i$ and $\oup_i$.
Among them, $\s{out}_i$ and $\oum_i$ denote the ledger-derived output metadata defined in \S\ref{subsec:ledgermodel};
the remaining transaction symbols denote immutable transaction data.
We further define $aexpr^\star$ analogously to $aexpr$ with an extended symbol set $S_a' \coloneqq S_a \cup \{\lambda_1, \lambda_2, \ldots\}$.
Each continuation marker $\lambda_i$ encodes a step-local parameter carried along recursive RI propagation.
The notation $a - b$ abbreviates $a + (-b)$.

A group expression $g \in gexpr$ is defined analogously:
\[
g \coloneqq
\theta
\mid g_1 \cdot g_2
\mid g^a
\mid \s{H}(e_1, \ldots, e_\ell)
\mid \inp_i g
\mid \oup_i g
\mid \Omega
\mid s_g
\]
where $\theta \in \mathbb{G}$, $s_g \in S_g$, $\ell \geq 1$, and
$
S_g \coloneqq \{ \s{pk}, \sigma_1, \sigma_2, \ldots, y_1, y_2, \ldots \}
$.
In a group-expression context, $\s{H}(e_1, \ldots, e_\ell)$ returns an element of $\mathbb{G}$, the ambient group for public keys, signatures, and group witnesses.
For simplicity, public keys and signatures are represented as single group elements.

Finally, Boolean expressions $b \in bexpr$ used in RIs are defined as:
\begin{align*}
b \coloneqq \s{U}
& \mid \bmbot
\mid \neg b
\mid b_1 \wedge b_2
\mid \s{Ver}(g_1, g_2)
\mid \s{SVer}(g_1, a, g_2)
\mid \inp_i b
\mid \oup_i b
\mid a_1 < a_2
\mid g_1 = g_2
\mid \s{HasIn}(i)
\\
& \mid \inp_i \s{RI}(a_1, a_2, \ldots)
\mid \oup_i \s{RI}(a_1, a_2, \ldots),
\end{align*}
where
$a, a_1, a_2, \ldots \in aexpr^\star$ and
$g_1, g_2 \in gexpr$.
The atom $g_1 = g_2$ is primitive equality over group expressions.
The predicate $\s{SVer}(g_1, a, g_2)$ represents signature verification, where $g_1$ is the public key, $a$ the message, and $g_2$ the signature.
$\s{Ver}(g_1, g_2)$ instead uses the hash of the current transaction excluding signatures as the message.
The atom $\inp_i/\oup_i \s{RI}(a_1, a_2, \ldots)$ asserts that the referenced transaction carries the same RI formula, with $\inp_i/\oup_i \lambda_j = a_j$ for each continuation argument.
The atom $\s{HasIn}(i)$ tests whether the current transaction has an $\ith$ input position;
unlike $\inp_i$, it does not dereference the predecessor transaction.

\itpara{RI Formulas versus $\s{RI}$ Matching Predicates.}
Each transaction carries one RI formula as data, written $\tx.\ri \in bexpr$.
In examples, $F \coloneqq \cdots$ names the RI formula stored in the current transaction.
In the DSL, $\s{RI}(a_1, \ldots, a_m)$, for $m \geq 0$, is instead a Boolean matching predicate: it checks that the referenced transaction carries the same RI formula and has exactly $m$ continuation markers, with marker $\lambda_j$ matching $a_j$ for every $j \in [m]$.
Accordingly, we use $\inp_i/\oup_i \s{RI}(\cdot)$ as RI-formula matching constraints.
When such a successor-matching constraint holds and the realizing transaction is accepted, that transaction carries the same RI formula and is checked under it;
we call this \emph{RI-formula propagation}.
With $m = 0$, $\s{RI}$ abbreviates $\s{RI}()$.

\itpara{Continuation Markers and Parameters.}
Each transaction stores continuation markers $\lambda_i$ associated with its RI formula;
their values may serve as step-local continuation parameters or state.
Recursive continuation uses matching predicates
$\oup_i \s{RI}(a_1, a_2, \ldots)$ whose arguments specify the required continuation-marker values.
For example, the formula
$(\lambda_1 < 100 \to \oup_1 \s{RI}(\lambda_1 + 1)) \wedge \oup_2 \bot$
requires, while $\lambda_1 < 100$, any transaction realizing successor position~$1$ to carry the same RI formula with continuation-marker value $\lambda_1 + 1$.
The continuation state is thus carried in the transaction's continuation markers associated with the RI formula rather than
encoded indirectly through ordinary transaction fields.

For Boolean expressions $\varphi, \psi \in bexpr$, we use the syntactic abbreviations $\varphi \to \psi \coloneqq \neg \varphi \vee \psi$ and $\varphi \vee \psi \coloneqq \neg(\neg \varphi \wedge \neg \psi)$.
Moreover,
$\bmtop \coloneqq \neg \bmbot$.
For arithmetic expressions $c, d \in aexpr^\star$, we write
$c = d \coloneqq \neg(c < d \vee d < c)$,
$c \ne d \coloneqq \neg(c = d)$,
$c \le d \coloneqq \neg(d < c)$, and
$c > d \coloneqq d < c$
(similarly for~$\ge$).
Equality of group expressions is the primitive atom $g_1 = g_2$ defined above.
Logic is interpreted under Kleene's three-valued semantics to accommodate future-dependent predicates during validation.
We now impose a syntactic restriction that prevents modal operands from depending on future output information.

\begin{definition}[Pure Expressions and Propositions]
\label{def:pure}
An arithmetic expression $A \in aexpr^\star$, a group expression $G \in gexpr$, or a proposition $t \in bexpr$ is \emph{pure} if it contains neither an output modality $\oup_i$ nor a ledger-derived output symbol $\s{out}_i$ or $\oum_i$, for any $i \in \mathbb{N}_{>0}$.
For propositions, this includes ordinary subformulas $\oup_i t$ and RI-matching atoms $\oup_i\s{RI}(\cdot)$;
the restriction applies recursively within nested expressions, hash arguments, and verification predicates.
\end{definition}

Under a successor modality, purity rules out output-of-output reasoning.
Under a predecessor modality, it prevents an already accepted transaction from observing output information of an ancestor that may be determined by a later ledger extension.


\begin{definition}[Evaluable Arithmetic and Group Expressions]
Validation-time evaluability for $aexpr^\star$ and $gexpr$ is defined inductively.
\begin{itemize}
\item A constant in $\mathbb{Z}$ (or $\mathbb{G}$), the symbol $\Omega$, any symbol in $S_a'$ (or $S_g$), the block height $\Lambda$, or the random value $\s{Rnd}$ is evaluable.

\item $-a$, $a + b$, $a \cdot b$, and $\lvert a \rvert$ are evaluable if $a$ and $b$ (if applicable) are evaluable.
The same holds for $gexpr$ for the applicable non-modal, non-hash operators.

\item $\s{H}(e_1, \ldots, e_\ell)$ is evaluable if each argument $e_j$ is evaluable according to its arithmetic or group sort.

\item
$\inp_i A$ and $\oup_i A$ are evaluable if $A$ is both pure and an evaluable arithmetic expression.
Likewise, $\inp_i G$ and $\oup_i G$ are evaluable if $G$ is both pure and an evaluable group expression.
\end{itemize}
\end{definition}

The evaluable fragment excludes output-of-output reasoning and future-sensitive ancestor-output dependencies, including through the ledger-derived symbols $\s{out}_j$ and $\oum_j$.
Based on evaluable arithmetic and group expressions, we define it below.

\begin{definition}[Validation-Time-Evaluable Recursive Invariant]
\label{def:validRI}
We define the validation-time evaluability of a recursive invariant formula
$F \in bexpr$ inductively.
\begin{itemize}
\item \emph{$F \coloneqq \s{U}$ or $F \coloneqq \bmbot$:}
$F$ is evaluable.

\item \emph{$F \coloneqq \neg F_L$:}
$F$ is evaluable if $F_L \in bexpr$ is evaluable.

\item \emph{$F \coloneqq \inp_i F_L$ or $F \coloneqq \oup_i F_L$:}
$F$ is evaluable if $F_L \in bexpr$ is both pure and evaluable.

\item \emph{$F \coloneqq F_L \wedge F_R$:}
$F$ is evaluable if $F_L, F_R \in bexpr$ are evaluable.

\item \emph{$F \coloneqq A_L < A_R$:}
$F$ is evaluable if $A_L, A_R \in aexpr^\star$
and $A_L$ and $A_R$ are evaluable arithmetic expressions.

\item \emph{$F \coloneqq G_L = G_R$:}
$F$ is evaluable if $G_L, G_R \in gexpr$
and $G_L$ and $G_R$ are evaluable group expressions.

\item \emph{$F \coloneqq \s{HasIn}(i)$:}
$F$ is evaluable.

\item \emph{$F \coloneqq \s{Ver}(G_L, G_R)$:}
$F$ is evaluable if $G_L, G_R \in gexpr$
and $G_L$ and $G_R$ are evaluable group expressions.

\item \emph{$F \coloneqq \s{SVer}(G_L, A_M, G_R)$:}
$F$ is evaluable if $G_L, G_R \in gexpr$ and $A_M \in aexpr^\star$,
and $G_L$ and $G_R$ are evaluable group expressions
and $A_M$ is an evaluable arithmetic expression.

\item \emph{$F \coloneqq \inp_i \s{RI}(A_1, A_2, \ldots)$:}
$F$ is evaluable if each $A_j \in aexpr^\star$
is an evaluable arithmetic expression.

\item \emph{$F \coloneqq \oup_i \s{RI}(A_1, A_2, \ldots)$:}
$F$ is evaluable if each $A_j \in aexpr^\star$
is an evaluable arithmetic expression.
\end{itemize}
\end{definition}

We write $\s{ValidRI}(F)$ when an RI formula $F$ is evaluable according to Definition~\ref{def:validRI}, matching the checker predicate.

\paragraph{Scope of parent-local re-checking.}
The validation-time locality claims assume RI-bearing transactions carry evaluable RI formulas and that their evaluations are field-well-formed: initial ones satisfy these conditions, and the checker checks each new RI for static evaluability and rejects out-of-range immutable fields.
Predecessor-side purity is essential;
otherwise, for example, $\inp_1 \oum_2 < c$ could change when an ancestor's second output is later realized even though the transaction carrying the clause is not a parent of that extension.

\begin{lemma}[Parent-local stability]
\label{lem:parent-local-stability}
Let $\Gamma^+ = \Gamma \cdot \tx$ be an accepted extension of a valid ledger and
$P = \{\tx.\s{in}_k \colon k \in [|\tx.\s{in}|]\}$.
If $q \in \Gamma$ carries an evaluable RI formula and $q.\id \notin P$, then
\[
\s{bEval}_{\Gamma^+, q}(q.\ri)=\s{bEval}_{\Gamma, q}(q.\ri).
\]
Hence only direct parents need RI re-checking.
\end{lemma}

\begin{proof}
Only direct parents gain new $\s{out}/\oum$ metadata;
existing predecessor relations, RI formulas, continuation markers, and inclusion contexts are immutable.
A structural induction then uses purity to show that ordinary modal operands cannot observe the changed metadata;
RI-matching arguments are stable by the same induction.
\end{proof}

The two exemplary RIs below illustrate the expressivity of \ctvp and are studied further in \S\ref{subsec:usecase}.

\itpara{Emulating Common UTXO Constraint Patterns.}
The DSL can encode common constraint patterns used by classical UTXO scripts and CTV-op.
Predicates $p_i$ (pure propositions) over the transaction realizing successor position~$i$ can be written as
$F \coloneqq \text{\acmdashbox{$\bigwedge_i \oup_i p_i$}}$.
Chain-shaped recursive invariants
(see the running example in \S\ref{sec:overview} and the collective payments and taxed transfer case studies in \S\ref{subsec:usecase})
naturally emulate EUTXO-style forward chaining.

\itpara{Multi-Step Flow.}
The DSL can also express distinct execution rules across steps.
Consider
$F \coloneqq \acmdashbox{$\oup_3 \bmbot\ \wedge\ \bigl( \Lambda \le 1000 \ \to\ \oup_1 \s{RI}(\lambda_1, \lambda_2, \ldots) \bigr) \ \wedge\ \s{pk} = \oup_1 \s{pk} \ \wedge\ \oum_2 < 100$}$.
While $\Lambda \le 1000$, any transaction realizing successor position~$1$ must carry the same RI formula.
The constraint $\oup_3 \bmbot$, together with gap-free output indexing, limits each step to at most two realized output positions.
The condition $\s{pk} = \oup_1 \s{pk}$ preserves the successor key, while $\oum_2 < 100$ bounds the second realized output amount if that position is realized.
Once $\Lambda > 1000$, RI propagation is no longer needed.
This supports workflows with staged rules, \eg, gaming, layer-two asset transfers, and temporal or cooldown-constrained flows (\S\ref{subsec:usecase}).

\subsection{Concrete Algorithms}
\label{subsec:algo}

Based on the block execution model (\S\ref{subsec:ledgermodel}) and the DSL for RIs (\S\ref{subsec:dsl}), we present transaction verification and ledger extension in Algorithm~\ref{algo:tx} and Algorithm~\ref{algo:ledger}.

\itpara{Data Structure.}
The internal transaction representation ($t \colon \textsc{Tx}$) augments immutable transaction data $\tx \in \mathcal{M}_{\tx}$ with a materialized representation of the ledger-derived output metadata.
Specifically, $t.\textrm{out}_j$ records the transaction spending output position $j$, while $t.\textrm{om}_j$ records the corresponding input amount;
neither is serialized in $\tx$.
To avoid duplicating RI text, the checker maintains a global repository $\mathsf{ISet}$ of RI formulas.
An incoming transaction carries its RI formula and continuation markers, and its internal representation receives a reference to the corresponding formula in the candidate repository.

The formula repository and fresh-label counter are checker-side state committed atomically with the ledger.
Ledger extension uses candidate versions for newly interned formulas and allocated RI-instance labels, so rejection leaves the committed checker state unchanged.

The concrete checker additionally associates each transaction with an RI-instance label.
This label is auxiliary checker-side metadata: it is not a component of the transaction model $\mathcal{M}_{\tx}$ in \S\ref{subsec:ledgermodel} and is
intentionally abstracted away by the denotational semantics of \S\ref{sec:dslri}.
Its purpose is to separate concurrent executions of the same RI formula in the concrete checker.

Specifically, the checker strengthens RI propagation with an instance-consistency condition.
A transaction that continues an existing RI instance reuses the label induced by that instance;
a transaction for which no such label is induced receives a fresh label.
If its inputs induce conflicting labels for the same RI formula, the checker rejects the transaction.

For an internal transaction $u$ with $u.\textrm{ri} = (q, c)$, write $\s{formula}(u) \coloneqq \star(q)$ and $\s{label}(u) \coloneqq c$.
For an internal transaction $t$ and RI formula $F$, define
\[
\s{IndLbl}_{F}(t)
\coloneqq
\left\{
\s{label}(\textrm{Ledger}[t.\textrm{in}_i])
\;\middle|\;
\begin{array}{l}
i \in [|t.\textrm{in}|],\\
\textrm{Ledger}[t.\textrm{in}_i] \neq \textrm{base},\\
\s{formula}(\textrm{Ledger}[t.\textrm{in}_i]) = F
\end{array}
\right\}.
\]
Under this conservative separation rule, an input whose parent carries the same RI formula induces that parent's instance label.
If $\s{IndLbl}_{F}(t) = \emptyset$, the transaction starts a fresh checker-side instance;
if it is a singleton, the transaction inherits that label;
otherwise, it is rejected.


\begin{algorithm}[t]
\caption{Transaction Verification}
\label{algo:tx}
\begin{algorithmic}[1]
\State \textbf{structure} \textsc{Tx}
\Comment{Augments $\mathcal{M}_{\tx}$ with materialized ledger-derived metadata}
\Statex \hspace{\algorithmicindent}$\textrm{id} \colon$ Integer
\Statex \hspace{\algorithmicindent}$\textrm{pk} \colon$ Group Element
\Statex \hspace{\algorithmicindent}$\textrm{x}, \textrm{in}, \textrm{im}, \textrm{om}, \textrm{oi}, \lambda \colon$ Vector of Integers
\Statex \hspace{\algorithmicindent}$\sigma, \textrm{y} \colon$ Vector of Group Elements
\Statex \hspace{\algorithmicindent}$\textrm{out} \colon$ Vector of (Integer or a symbol $\Omega$)
\Comment{$\textrm{om}$ and $\textrm{out}$ materialize ledger-derived metadata}
\Statex \hspace{\algorithmicindent}$\textrm{ri} \colon$ Pair of $\star bexpr$ and Integer
\Comment{RI-formula pointer and checker-side instance label}
\State \textbf{end structure}
\State $\textrm{Ledger} \colon$ Map from Integer to \textsc{Tx} \Comment{Global variables}
\State $\textrm{ISet} \colon$ Set of $bexpr$
\State $\textrm{CtRI} \colon$ Integer

\Procedure{VerifyTx}{$\Gamma \colon \textrm{Ledger},\ t \colon \textsc{Tx}$}
\State $m_i \gets |t.\textrm{in}|$, $m_o \gets |t.\textrm{om}|$

\IfReturn{$\exists i \in [m_i] \colon t.\textrm{im}_i < 0
	\;\lor\;
	\exists i \in [m_o] \colon t.\textrm{om}_i < 0$}
	{$\s{False}$}
\Comment{Non-negative}

\If{$\Gamma[t.\textrm{in}_1] = \textrm{base}$}
	\IfReturn{$\neg \s{coinBaseEval}(t)$}
		{$\s{False}$}
	\State \Return $\s{CheckRI}(\Gamma, t, \s{formula}(t))$
\EndIf
\Comment{Coinbase issuance and RI admissibility}

\IfReturn{$
	\sum_{i \in [m_i]} t.\textrm{im}_i
	<
	\sum_{i \in [m_o]} t.\textrm{om}_i
$}
	{$\s{False}$}
\Comment{Checks balance consistency}

\State \Return $\s{CheckRI}(\Gamma, t, \s{formula}(t))$
\Comment{Checks RI semantics and checker-side label consistency}
\EndProcedure
\end{algorithmic}
\end{algorithm}

\itpara{Transaction Verification.}
Algorithm~\ref{algo:tx} checks non-negative amounts and balance consistency, \ie,
that the total input value is at least the total value of the currently realized outputs.
It then evaluates the RI via \textsc{CheckRI}, which is the implementation-level realization of the semantic evaluator $\s{bEval}$ in the current ledger context and follows the admissibility relation $\models$ defined in \S\ref{subsec:semantics}.
It rejects out-of-range immutable fields.
For RI-formula and continuation-marker matching, \textsc{CheckRI} follows the denotational matcher of \S\ref{subsec:semantics} and additionally enforces the checker-side instance-label consistency condition described above.
We omit the unfolding of \textsc{CheckRI} and use Algorithm~\ref{algo:tx} as a callable subroutine in ledger extension.

\itpara{Ledger Extension.}
Algorithm~\ref{algo:ledger} parses an incoming transaction into $t$, initializes $out/om$ metadata as empty, and checks format, identifier uniqueness, parent existence, and fresh gap-free realized output positions.
The distinguished $\textrm{base}$ reference is existence-checked but excluded from ordinary parent realization and re-checking.



\begin{algorithm}
\caption{Ledger Extension}
\label{algo:ledger}
\begin{algorithmic}[1]
\Procedure{LedgerExt}{$\tx \colon \mathcal{M}_{\tx}$}
	\State Initialize the immutable fields of $t \colon \textsc{Tx}$ from $\tx$;
	$t.\textrm{out}, t.\textrm{om} \gets \epsilon$
	\State $m \gets |t.\textrm{in}|$
	\IfReturn{$m \ne |t.\textrm{im}| \lor m \ne |t.\textrm{oi}|$}{}
	\Comment{Aligned input tuples}
	\IfReturn{$t.\textrm{id} \in \s{lift}(\textrm{Ledger}, \textrm{id})
	\lor t.\textrm{pk} \ne t.\textrm{y}_1$}{}
	\Comment{Format and (o4)}
	\IfReturn{$\neg \s{ValidRI}(\tx.\ri)$}{}
	\Comment{Static RI-evaluability check}

	\State $P_{\mathrm{all}} \gets \{t.\textrm{in}_i \colon i \in [m]\}$
	\IfReturn{$\exists p \in P_{\mathrm{all}} \colon p \notin \s{lift}(\textrm{Ledger}, \textrm{id})$}{}
	\State $P \gets \{p \in P_{\mathrm{all}} \colon \textrm{Ledger}[p] \ne \textrm{base}\}$
	\For{$p \in P$}
		\State $I_p \gets \{i \in [m] \colon t.\textrm{in}_i = p\}$,
		$\ell_p \gets |\textrm{Ledger}[p].\textrm{om}|$,
		$J_p \gets \{t.\textrm{oi}_i \colon i \in I_p\}$

		\IfReturn{$|J_p| \ne |I_p|
		\lor J_p \ne \{\ell_p + 1, \ldots, \ell_p + |I_p|\}$}{}
		\Comment{Fresh, gap-free outputs}
	\EndFor

	\State $\textrm{CtRI}' \gets \textrm{CtRI}$;
	$\textrm{ISet}' \gets \textrm{ISet}$
	\Comment{Candidate formula repository \textrm{ISet}$'$}
	\State Let $q$ point to $\tx.\ri$ in $\textrm{ISet}'$, inserting $\tx.\ri$ if absent

	\State $C \gets \s{IndLbl}_{\tx.\ri}(t)$
	\State \algorithmicif\ $C = \emptyset$
	\algorithmicthen\ $c \gets \textrm{CtRI}' + 1$;
	$\textrm{CtRI}' \gets c$
	\State \algorithmicelse\ \algorithmicif\ $|C| = 1$
	\algorithmicthen\ $c \gets$ the unique element of $C$
	\State \algorithmicelse\ \algorithmicreturn
	\State \algorithmicend\ \algorithmicif
	\State $t.\textrm{ri} \gets (q, c)$;
%
	$\textrm{Ledger}' \gets \textrm{Ledger}$
	\Comment{Candidate ledger state}
	\ForDo{$i \in [m] \colon \textrm{Ledger}[t.\textrm{in}_i] \ne \textrm{base}$}
		{$\textrm{Ledger}'[t.\textrm{in}_i].(\textrm{out}_{t.\textrm{oi}_i},
		\textrm{om}_{t.\textrm{oi}_i})
		\gets (t.\textrm{id}, t.\textrm{im}_i)$}
	\State append $t$ to $\textrm{Ledger}'$

	\IfReturn{$\neg \textsc{VerifyTx}(\textrm{Ledger}', t)
	\lor \exists p \in P \colon
	\neg \textsc{VerifyTx}(\textrm{Ledger}', \textrm{Ledger}'[p])$}{}
	\State $(\textrm{Ledger}, \textrm{ISet}, \textrm{CtRI}) \gets (\textrm{Ledger}', \textrm{ISet}', \textrm{CtRI}')$
	\Comment{Atomic commit}
\EndProcedure
\end{algorithmic}
\end{algorithm}

The procedure creates a candidate formula repository $\textrm{ISet}'$, interns $t$'s RI formula if needed, assigns its RI reference, and computes $C \gets \s{IndLbl}_{\tx.\ri}(t)$.
If $C = \emptyset$, $t$ starts a fresh checker-side RI instance;
if $|C| = 1$, it inherits the unique induced label;
otherwise, it is rejected.

The candidate ledger then associates each newly realized parent output position with $t$ and the corresponding input amount.
The resulting $out$ and $om$ entries materialize ledger-derived metadata in the checker state;
they are not immutable parent data.
The spender identifier and realized amount arise only when a later transaction realizes a referenced successor position, which the parent's RI constrains.
The checker evaluates RI and balance conditions on the candidate state, verifies $t$, and re-checks each affected parent.
By Lemma~\ref{lem:parent-local-stability}, no wider re-validation is needed.
On failure, the candidate state is discarded;
otherwise, the checker state is atomically committed.

\itpara{Relationship to the formal model.}
For an ordinary transaction, successful completion of Algorithm~\ref{algo:ledger} implies the corresponding acceptance conditions of \S\ref{subsec:ledgermodel}: the algorithm checks input well-formedness, identifier uniqueness, non-negative amounts, balance consistency, the candidate RI, and the RIs of all affected parents on the candidate ledger state.
The checker additionally enforces the evaluable-fragment restriction of Definition~\ref{def:validRI} and the checker-side RI-instance separation condition described above.
Consequently, checker acceptance refines the formal model;
the converse need not hold because these checker-side conditions can reject a transaction admitted by the formal model.
Coinbase issuance is handled separately by the host-ledger-specific rule described below.

\itpara{On Coinbase Transactions.}
A coinbase transaction mints coins under the issuance rule of the underlying blockchain.
To preserve the transaction interface, the initial ledger contains a distinguished base sentinel whose reference identifies coinbase transactions.
A transaction is treated as coinbase when its first input references this sentinel;
Algorithm~\ref{algo:ledger} therefore neither mutates nor re-checks the base entry.
Algorithm~\ref{algo:tx} uses host-ledger-specific \textsc{coinBaseEval} in place of the ordinary balance check, then checks RI admissibility;
\textsc{coinBaseEval} is responsible for the host ledger's complete issuance rule, including the permitted coinbase input shape.

\itpara{On Concurrency and Failure Handling.}
We model the ledger as a serialized, append-only transaction log and interpret RIs as safety properties over its canonical history.
If a transaction depends on a predecessor that has not been included in the ledger, it is rejected by the parent-existence check.
Concurrent proposals may exist at the network or mempool level, but consensus commits at most one conflicting realization, and the model reasons only about the resulting canonical sequence.

Ledger extension is modeled as an atomic transition from the committed ledger to a candidate ledger state.
We retain a serialized canonical history and a single-realization discipline for output positions, while \ctvp augments validation with ledger-derived output metadata.
Algorithm~\ref{algo:ledger} validates the candidate state and commits it only if all checks succeed;
otherwise, the committed ledger is unchanged.

Handling liveness and coordination, such as retries, rollbacks, timeouts, and refunds, is orthogonal to RI safety and is typically enforced by the coordination layer in deployed protocols.
When explicit on-chain failure handling is needed, it can be encoded by adding RI cases, \eg, keyed on height, time, or output tags, at the cost of longer specifications and additional branching.
These coordination costs remain;
the associated RI safety checks retain \ctvp's validation-time locality.

\itpara{Integration Requirements for Bitcoin-like Ledgers.}
Applying our model to a Bitcoin-like system requires consensus-visible changes to transaction encoding and validation.
Transactions must carry an RI formula and continuation markers, validators must invoke the RI interpreter, and nodes must maintain or reconstruct ledger-derived output metadata.
These metadata are deterministic ledger functions, not application-visible mutable contract state.

Deployment further requires RI-aware parsing, relay, mempool policy, wallets, and tooling.
Clauses using $\Lambda$ or $\s{Rnd}$ take their values from the candidate inclusion context and reuse the committed values on re-checking.
Thus, \ctvp changes consensus-critical encoding and validity rules, but not transaction ordering or finality.

\subsection{Expressiveness: Turing-Completeness}
\label{subsec:turing}
Turing-completeness here concerns unbounded mathematical integers and ledger extensions, not the finite $256$-bit benchmark workloads of \S\ref{sec:evaluation}.
We simulate a UTXO counter machine~\cite{mst/FischerMR68,csf/BartolettiLZ21}, following neighborhood covenants~\cite{csf/BartolettiLZ21};
all RI formulas used satisfy Definition~\ref{def:validRI}.

A machine consists of $n$ counters and a finite instruction sequence $s = (s_1, \ldots, s_{|s|})$.
A machine state is $(v_1, \ldots, v_n, p)$, where each $v_i \in \mathbb{N}$ is a counter value and $p$ is the program counter.
Instructions are of the form $\s{inc}\ i$, $\s{dec}\ i$, $\s{zero}\ i$, $\s{ifGoto}\ i j$, and $\s{halt}$.
For the operational semantics used here, $\s{inc}\ i$ increments $v_i$ and advances to the next instruction;
$\s{dec}\ i$ decrements $v_i$ and advances only when $v_i > 0$ (and has no successor when $v_i = 0$);
$\s{zero}\ i$ sets $v_i$ to $0$ and advances;
$\s{ifGoto}\ i j$ jumps to $j$ when $v_i \ne 0$ and otherwise advances;
and $\s{halt}$ terminates.
We assume that every jump target $j$ is in $[|s|]$ and that every non-halting instruction either has a valid successor at position $k+1$ or the program is padded so that $k + 1 \in [|s|]$.

\begin{theorem}
\ctvp simulates any UTXO counter machine.
Since UTXO counter machines are Turing-complete, \ctvp is Turing-complete.
\end{theorem}

\begin{proof}
Fix distinct key values $A, B \in \mathbb{G}$.
Fix a UTXO counter machine with $n$ counters, instruction sequence $s$, and initial balance~$R$.
We encode one machine state per transaction.
For each $i \in [n]$, field $x_i$ stores counter value $v_i$, and field $x_{n + 1}$ stores the program counter $p$.
The RI formula is propagated through successor position~$1$, whose realized amount is constrained to~$R$ throughout the computation.
At halt, RI propagation stops;
a realized successor has amount $R$ and key field $A$ if $x_1 = 0$, and $B$ otherwise.
We define:
\[
F \coloneqq
\left(\bigwedge_{i = 1}^{n} x_i \ge 0\right)
\ \wedge\
(1 \le x_{n + 1}) \wedge (x_{n + 1} \le |s|)
\ \wedge\
\oum_1 = R
\ \wedge\
\oup_2 \bmbot
\ \wedge\
\bigwedge_{k = 1}^{|s|}\left(x_{n + 1} = k \to \s{cond}_k\right).
\]
Here $\s{cond}_k$ encodes the transition prescribed by instruction $s_k$.
If $s_k$ is $\s{inc}\ i$, let
\[
\s{cond}_k
\coloneqq
\left(\bigwedge_{t \in [n] \setminus \{i\}} x_t = \oup_1 x_t\right)
\ \wedge\
\oup_1 x_i = x_i + 1
\ \wedge\
\oup_1 x_{n + 1} = k + 1
\ \wedge\
\oup_1 \s{RI}.
\]

If $s_k$ is $\s{dec}\ i$, let
$
\s{cond}_k
\coloneqq
\left(\bigwedge_{t \in [n] \setminus \{i\}} x_t = \oup_1 x_t\right)
\ \wedge\
x_i = \oup_1 x_i + 1
\ \wedge\
\oup_1 x_{n + 1} = k + 1
\ \wedge\
\oup_1 \s{RI}
$.

If $s_k$ is $\s{zero}\ i$, let
$
\s{cond}_k
\coloneqq
\left(\bigwedge_{t \in [n] \setminus \{i\}} x_t = \oup_1 x_t\right)
\ \wedge\
\oup_1 x_i = 0
\ \wedge\
\oup_1 x_{n + 1} = k + 1
\ \wedge\
\oup_1 \s{RI}
$.

If $s_k$ is $\s{ifGoto}\ i\,j$, let
\[
\s{cond}_k
\coloneqq
\left(x_i \ne 0 \to \oup_1 x_{n + 1} = j\right)
\ \wedge\
\left(x_i = 0 \to \oup_1 x_{n + 1} = k + 1\right)
\ \wedge\
\left(\bigwedge_{t \in [n]} x_t = \oup_1 x_t\right)
\ \wedge\
\oup_1 \s{RI}.
\]

If $s_k$ is $\s{halt}$, let
$
\s{cond}_k
\coloneqq
\left(x_1 = 0 \to \oup_1 y_1 = A\right)
\ \wedge\
\left(x_1 \ne 0 \to \oup_1 y_1 = B\right)
$.

The initial transaction sets $x_i = 0$ for each $i \in [n]$ and $x_{n + 1} = 1$.
At each non-halting state, the current program-counter value selects exactly one instruction clause, so any realized successor encodes the prescribed transition, preserves $R$, and propagates the RI.
Conversely, each enabled machine transition has a realizing successor;
at halt the final branch is enforced without further propagation.
Hence every machine execution is realizable as an RI-governed chain, and every realized non-halting continuation is prescribed;
RIs constrain rather than force continuation (\S\ref{subsec:limitations}).
\end{proof}


\section{Semantics and Soundness}
\label{sec:dslri}
We instantiate the DSL semantics for the evaluable fragment of Definition~\ref{def:validRI} (\S\ref{subsec:semantics}) and present proof rules for RI reasoning (\S\ref{subsec:syntax}).
The key validation-time feature is that future-dependent clauses evaluate to $\mathsf{unknown}$ until the referenced successor position is realized, then become checkable.
We prove the deduction system sound for this semantics on the evaluable fragment.

\subsection{Semantics}
\label{subsec:semantics}

This subsection fixes the denotational semantics used by the soundness theorem, focusing on future-spend modalities when the referenced transaction is absent.
\ifnum\fullFlag = 1
We give the modal clauses here and defer the routine structural recursion for arithmetic, group, and Boolean connectives to Appendix~\ref{app:eval}.
\else
We give the modal clauses here and defer the routine structural recursion for arithmetic, group, and Boolean connectives to Appendix~A.2 of the full version~\cite{eprint/TangCFGL26}.
\fi
For a transaction $\tx$, we write $(\Gamma, \tx)\models t$ for validation admissibility of $t\in bexpr$ under a ledger $\Gamma\in\mathfrak{M}$.
We model a missing reference by a distinguished transaction $\bmbot_{\tx}$ and stipulate $(\Gamma, \bmbot_{\tx})\not\models t$ for all $t\in bexpr$.

Write $\s{FieldWF}_{\Gamma, \tx}(p)$ when every immutable indexed terminal reached while structurally evaluating $p$ from $(\Gamma, \tx)$ is present in the transaction against which it is evaluated;
if a modal target is absent, its operand is not traversed for this check.
The denotational evaluators below are applied only when this field-well-formedness condition holds, and the checker rejects violations.
Only unresolved modal targets and unrealized $\s{out}_j/\oum_j$ metadata produce unknown.
The evaluation functions $\aeval$ and $\geval$ are defined by structural recursion in
\ifnum\fullFlag = 1
Appendix~\ref{app:eval}
\else
Appendix~A.2
\fi
and return $\Omega$ whenever a required subexpression does.
Immutable transaction symbols are read from $\tx$;
$\s{out}_j$ and $\oum_j$ use the ledger-derived metadata of
\S\ref{subsec:ledgermodel} and evaluate to $\Omega$ until the $\jth$ output
position is realized.
The symbols $\Lambda$ and $\s{Rnd}$ use the inclusion context fixed for $\tx$,
so later re-checking does not change their values.

\begin{definition}[Boolean Expression Evaluation: Core Clauses]
\label{def:core_clauses}
Let $\neg$, $\wedge$, and $\vee$ denote the Kleene three-valued connectives.
The evaluator
$\beval \colon bexpr \to \{\bmbot, \s{U}, \bmtop\}$
is defined by the standard propositional clauses and the modal/RI clauses below.

\itpara{Abstraction from RI-instance labels.}
The denotational semantics intentionally abstracts from the checker-side RI-instance labels introduced in \S\ref{subsec:algo}.
These labels are not components of
$\mathcal{M}_{\tx}$ and do not affect $\mathsf{aEval}$, $\mathsf{gEval}$, $\mathsf{bEval}$, or the admissibility relation $\models$.
Accordingly, the matcher below compares only the RI formula and continuation-marker values.

The concrete checker of \S\ref{subsec:algo} additionally enforces instance-label consistency;
any accepted RI match therefore satisfies the label-abstracted matcher below, but the converse can fail when labels conflict.

For a referenced transaction $\tx'$ and continuation arguments
$a_1, \ldots, a_m$, let
$v_j \coloneqq \aeval(a_j)$ for $j \in [m]$.
Define
\[
\s{riMatch}_{\Gamma, \tx}(\tx'; a_1, \ldots, a_m)
\coloneqq
\begin{cases}
\bmbot
& \text{if } \tx'.\ri \neq \tx.\ri
\text{ or } |\tx'.\lambda| \neq m,\\

\bmbot
& \text{if } \exists j \in [m] \colon
  v_j \neq \Omega \wedge \tx'.\lambda_j \neq v_j, \\

\s{U}
& \text{if } \exists j \in [m] \colon v_j = \Omega, \\

\bmtop
& \text{otherwise.}
\end{cases}
\]

\begin{align*}
\beval(\s{HasIn}(i))
&\coloneqq
\begin{cases}
\bmtop
& \text{if } i \in [|\tx.\s{in}|], \\
\bmbot
& \text{otherwise},
\end{cases}
\\
\beval(\inp_i \varphi)
&\coloneqq
\begin{cases}
\s{bEval}_{\Gamma, \tx'}(\varphi)
& \text{if } \s{prev}_i(\tx, \tx')
	\text{ for some } \tx' \in \Gamma, \\
\s{U}
& \text{otherwise},
\end{cases}
\\
\beval(\oup_i \varphi)
&\coloneqq
\begin{cases}
\s{bEval}_{\Gamma, \tx'}(\varphi)
& \text{if } \s{succ}^{\Gamma}_i(\tx, \tx')
	\text{ for some } \tx' \in \Gamma, \\
\s{U}
& \text{otherwise},
\end{cases}
\end{align*}
\begin{align*}
\beval(a_1 < a_2)
&\coloneqq
\begin{cases}
\s{U}
& \!\text{if } \Omega \in \{\aeval(a_1), \aeval(a_2)\}, \\
\aeval(a_1) < \aeval(a_2)
& \!\text{otherwise}.
\end{cases}
\\
\beval(g_1 = g_2)
&\coloneqq
\begin{cases}
\s{U}
& \text{if } \Omega \in \{\geval(g_1), \geval(g_2)\}, \\
\bmtop
& \text{if } \geval(g_1) = \geval(g_2), \\
\bmbot
& \text{otherwise},
\end{cases}
\end{align*}
\begin{align*}
\beval(\inp_i \s{RI}(a_1, \ldots, a_m))
&\coloneqq
\begin{cases}
\s{riMatch}_{\Gamma, \tx}(\tx'; a_1, \ldots, a_m)
& \text{if } \s{prev}_i(\tx, \tx')
	\text{ for some } \tx' \in \Gamma, \\
\s{U}
& \text{otherwise},
\end{cases}
\\
\beval(\oup_i \s{RI}(a_1, \ldots, a_m))
&\coloneqq
\begin{cases}
\s{riMatch}_{\Gamma, \tx}(\tx'; a_1, \ldots, a_m)
& \text{if } \s{succ}^{\Gamma}_i(\tx, \tx')
	\text{ for some } \tx' \in \Gamma, \\
\s{U}
& \text{otherwise}.
\end{cases}
\end{align*}
\end{definition}

For the verification clauses, write
$
\widehat{g}_i \coloneqq \geval(g_i),
\widehat{a}_1 \coloneqq \aeval(a_1)
$.
Then
\[
\beval(\s{Ver}(g_1, g_2))
\coloneqq
\begin{cases}
\s{U}
& \text{if } \Omega \in \{\widehat{g}_1, \widehat{g}_2\}, \\
\s{SigVer}'(\widehat{g}_1, M, \widehat{g}_2)
& \text{otherwise,}
\end{cases}
\]
and
\[
\beval(\s{SVer}(g_1, a_1, g_2))
\coloneqq
\begin{cases}
\s{U}
& \text{if } \Omega \in
	\{\widehat{g}_1, \widehat{a}_1, \widehat{g}_2\}, \\
\s{SigVer}(\widehat{g}_1, \widehat{a}_1, \widehat{g}_2)
& \text{otherwise.}
\end{cases}
\]

Here,
$\s{prev}_i(\tx_1, \tx_2)$ abbreviates
$\tx_1.\s{in}_i = \tx_2.\id$,
and $\s{succ}^{\Gamma}_i(\tx_1, \tx_2)$ holds iff the ledger-derived metadata of $\tx_1$ in $\Gamma$ satisfies $\tx_1.\s{out}_i = \tx_2.\id$.
Thus, successor lookup is ledger-relative even though the underlying transactions are immutable.
In $\s{Ver}(g_1, g_2)$,
$M$ is the hash of the immutable serialized data of the current transaction, excluding signatures and ledger-derived metadata such as $\s{out}$ and $\oum$.
The inclusion-context values $\Lambda$ and $\s{Rnd}$ are likewise not part of $M$.
The predicates
$\s{SigVer}'(pk, M, \sigma)$
and $\s{SigVer}(pk, m, \sigma)$
verify $\sigma$ under public key $pk$
for the implicit message $M$ and explicit message $m$, respectively.
Boolean comparison/verification results are identified with $\bmtop/\bmbot$.
A missing reference under $\inp_i$ or $\oup_i$ yields $\s{U}$ rather than immediate failure;
for an unrealized successor under $\oup_i$, this records a deferred obligation.
$\s{HasIn}(i)$ instead tests current input arity strictly.

\begin{definition}[Admissibility Relation]
The admissibility relation
$\models \subseteq
(\mathfrak{M} \times \mathcal{M}_{\tx}) \times bexpr$
is defined by
\[
\Gamma, \tx \models p
\quad\text{iff}\quad
\tx \neq \bmbot_{\tx},
\quad
\s{FieldWF}_{\Gamma, \tx}(p),
\quad\text{and}\quad
\beval(p) \in \nonf.
\]
\end{definition}

Thus, $\models$ denotes validation admissibility, not classical truth:
$\bmtop$ and $\s{U}$ are designated, while $\bmbot$ rejects;
$\s{U}$ is an admissible unknown, and for an unrealized successor it records a deferred obligation.
\subsection{Deduction System and Soundness}
\label{subsec:syntax}

We adopt a judgmental deduction system $\mathcal{N}$.
Judgments have the form $\Gamma, \tx \vdash \varphi$, where $\Gamma$ is a ledger, $\tx$ is a transaction, and $\varphi$ is a formula.
We refer to $(\Gamma, \tx)$ as the evaluation context and to $\varphi$ as the conclusion.
A derivation is \emph{field-well-formed} if every judgment $\Gamma, \tx \vdash \varphi$ occurring in it satisfies $\s{FieldWF}_{\Gamma, \tx}(\varphi)$;
the deduction system is considered only on such derivations.
If $\Gamma, \tx \vdash \varphi$ holds for every $\tx\in\Gamma$, we write $\Gamma \vdash \varphi$ (and similarly for $\Gamma \models \varphi$).

Figure~\ref{fig:pfrl} lists the rule schemata used in the paper.
Most propositional rules are standard for the underlying three-valued connective semantics.
The RI-specific content is modal evaluation over predecessor/successor transactions, including absent references yielding $\s{U}$;
we claim soundness, not completeness.

\begin{figure}[!ht]
\centering
\begin{align*}
\begin{array}{c c}
\prftree[r]{\text{($\wedge$-introduction)}}{\Gamma, \tx \vdash \varphi}{\Gamma, \tx \vdash \psi}{\Gamma, \tx \vdash \varphi \wedge \psi} &
\prftree[double][r]{\text{($\to$-introduction)}}{\Gamma, \tx \vdash \neg \varphi \vee \psi}{\Gamma, \tx \vdash \varphi \to \psi} \\
\prftree[r]{\text{($\wedge$-elimination left)}}{\Gamma, \tx \vdash \varphi \wedge \psi}{\Gamma, \tx \vdash \varphi} &
\prftree[r]{\text{($\wedge$-elimination right)}}{\Gamma, \tx \vdash \varphi \wedge \psi}{\Gamma, \tx \vdash \psi} \\
\prftree[r]{(excluded middle)}{}{\Gamma, \tx \vdash \varphi \vee \neg \varphi} &
\prftree[r]{(ex falso quodlibet)}{\Gamma, \tx \vdash [\rhd] \bmbot}{\Gamma, \tx \vdash [\rhd]\varphi} \\
\prftree[r]{\text{($\vee$-introduction left)}}{\Gamma, \tx \vdash \varphi}{\Gamma, \tx \vdash \varphi \vee \psi} &
\prftree[r]{\text{($\vee$-introduction right)}}{\Gamma, \tx \vdash \psi}{\Gamma, \tx \vdash \varphi \vee \psi} \\
\multicolumn{2}{c}{\prftree[r]{(case analysis)}{\Gamma, \tx \vdash \varphi \to \psi}{\Gamma, \tx \vdash \varphi' \to \psi}{\Gamma, \tx \vdash \varphi \vee \varphi' \to \psi} } \\
\multicolumn{2}{c}{\prftree[r]{(De Morgan)}{\Gamma, \tx \vdash \neg(\varphi_1 \wedge \varphi_2 \wedge \ldots \wedge \varphi_\ell)}{\Gamma, \tx \vdash \neg\varphi_1 \vee \neg\varphi_2 \vee \ldots \vee \neg\varphi_\ell} } \\
\multicolumn{2}{c}{\prftree[r]{(distribution)}{\Gamma, \tx \vdash (\varphi_1 \vee \varphi_2) \wedge \psi}{\Gamma, \tx \vdash (\varphi_1 \wedge \psi) \vee (\varphi_2 \wedge \psi)} } \\
\prftree[double][r]{(disjunction distribution, \text{$\vee$}D)}{\Gamma, \tx \vdash \rhd(\varphi \vee \psi)}{\Gamma, \tx \vdash \rhd \varphi \vee \rhd \psi} &
\prftree[double][r]{(K-distribution, K)}{\Gamma, \tx \vdash \rhd(\varphi \to \psi)}{\Gamma, \tx \vdash \rhd \varphi \to \rhd \psi} \\
\prftree[double][r]{(conjunction distribution, \text{$\wedge$}D)}{\Gamma, \tx \vdash \rhd(\varphi \wedge \psi)}{\Gamma, \tx \vdash \rhd \varphi \wedge \rhd \psi} &
\prftree[r]{(tautology, T)}{\Gamma, \tx \vdash \varphi_1 \to \varphi_2 \to \psi}{\Gamma, \tx \vdash \varphi_1 \wedge \varphi_2 \to \psi}\\
\end{array}
\end{align*}
\CaptionAndDescription{The deduction system $\mathcal{N}$ ($\rhd$ stands for $\inp_i$ or $\oup_i$, and $[\rhd]$ stands for rules that hold with or without $\rhd$)}
\label{fig:pfrl}
\end{figure}

Soundness states that every derivation in $\mathcal{N}$ is admissible under the denotational semantics of \S\ref{subsec:semantics}.
We give a proof sketch below;
the complete rule-by-rule proof appears in
\ifnum\fullFlag = 1
Appendix~\ref{app:detailproof}.
\else
Appendix~D of the full version~\cite{eprint/TangCFGL26}.
\fi
The artifact also includes an auxiliary Coq development of the three-valued modal logic~\cite{tool/coq}.\footnote{%
See \url{https://github.com/htftsy/CTVPlus/blob/main/CoqVrf/innerLogic.v}, which also proves additional derived rules.}

\begin{theorem}[Soundness]
For any valid ledger $\Gamma$, transaction $\tx$, and $p \in bexpr$, if there is a field-well-formed derivation of $\Gamma, \tx \vdash p$, then $\Gamma, \tx \models p$.
\end{theorem}

This theorem concerns the label-abstracted denotational semantics of \S\ref{subsec:semantics}.
The checker-side RI-instance separation condition of \S\ref{subsec:algo} is an additional checker-side restriction and is not part of the denotational relation proved sound here.

\begin{proof}
Fix a valid ledger $\Gamma$ and transaction $\tx$.
We induct on the last rule deriving $\Gamma, \tx\vdash p$.
Field well-formedness is assumed, so it suffices to show that the conclusion evaluates non-false.
Propositional cases follow from the Kleene truth tables.
For $\wedge$-intro, the induction hypotheses give
$\beval(\varphi), \beval(\psi) \in \nonf$;
hence $\beval(\varphi\wedge\psi) \in \nonf$, as required.

The only nontrivial cases are the modal distribution rules (K, $\vee$D, and $\wedge$D).
If the referenced $\ith$ predecessor or successor transaction is absent, both sides evaluate to $\s{U}$ by Definition~\ref{def:core_clauses}.
Otherwise, the rule reduces to the corresponding propositional identity in the referenced transaction and follows from the induction hypothesis.
The other cases are analogous.
\end{proof}


\section{Evaluation}
\label{sec:evaluation}
We use \emph{workflow} for an abstract protocol pattern,
\emph{case study} for its RI instantiation, and \emph{workload} for a measured transaction trace.
We evaluate \ctvp along two dimensions.
\S\ref{subsec:usecase} instantiates six practice-motivated workflows as RI case studies;
Fig.~\ref{fig:representative} depicts five, with the NFT case specified textually, and Table~\ref{tab:anchors} summarizes all six.
\S\ref{subsec:comparison} compares the benchmark proxy with Solidity references:
five workloads are trace-aligned, while the cooldown uses a linear reference.


\subsection{Case Studies}
\label{subsec:usecase}
We now detail the six RI workflows within the evaluable one-step fragment.
For each case, we give the relevant transaction fields, RI constraints, branch semantics, and the obligations checked immediately or carried forward.
Table~\ref{tab:anchors} links each workflow obligation to a motivating contract family, protocol, standard, or workflow pattern;
these anchors motivate the obligation but do not imply deployment of the RI encoding.
The evaluation isolates on-chain enforcement and validation-cost accounting;
off-chain dissemination of witnesses or application data is outside scope.
\ifnum\fullFlag = 1
Appendix~\ref{app:deploymentanchors}
\else
Appendix~B.2 of the full version~\cite{eprint/TangCFGL26}
\fi
describes the Solidity baselines and workload mappings.


\begin{table*}[t]
\centering
\caption{Real-world anchors for the six case studies}
\label{tab:anchors}
\setlength{\arrayrulewidth}{0.1pt}
\setlength{\tabcolsep}{3.5pt}
\begin{tabular}{p{0.17\textwidth} p{0.40\textwidth} p{0.37\textwidth}}
\toprule
Case study &
Anchors in practice &
What our RI encoding captures
\\
\midrule
Non-fungible tokens
&
ERC-721 defines the NFT ownership/transfer interface, and OpenZeppelin follows this ownership pattern~\cite{EIP721,OZContracts}.
&
Ownership transition:
minting creates ownership;
transfer preserves the token identifier and updates the owner.
\\
\hline
Gaming with guesses and an external beacon
&
Solidity exposes \texttt{blockhash} and \texttt{prevrandao}, with documented manipulation risks~\cite{SolidityGlobals,EIP4399};
dApps also use verifiable-randomness services such as Chainlink~\cite{ChainlinkVRF}.
&
Two-player guessing with public randomness and RI-carried game state.
\\
\hline
Layer-two assets (split/merge)
&
RGB and Taproot Assets use UTXO-anchored split/merge state transitions~\cite{RGBClientSideValidation24,RGBSingleUseSeals24,RGBStateTransitions25,LightningTaprootAssetsProtocol24,LightningTaprootAssetsOverview24}.
&
UTXO-anchored asset transitions with RI-enforced validity constraints.
\\
\hline
Collective payments
&
Uniswap's Merkle Distributor authenticates recipient claims against a committed Merkle root~\cite{UniswapMerkleDistributor}.
&
Sequential authorized claims with an algebraic membership-witness state and a one-unit payout condition.
\\
\hline
Transfer-tax token
&
Fee-on-transfer tokens are common;
routers support them (\eg, \texttt{SupportingFeeOnTransfer}) and document compatibility issues~\cite{UniswapRouter02,UniswapCommonErrors,UniswapFoTIssue}.
&
Fixed $1{:}9$ relation between realized fee and transfer amounts, with RI-carried continuation state.
\\
\hline
Branch-local cooldown flow
&
Smart-account wallets use spending-cap and allowance modules (\eg, Safe)~\cite{SafeSpendingLimits,SafeAllowanceModule}.
&
RI-encoded per-branch cooldown with a parent-to-child inclusion-gap threshold until expiry.
\\
\bottomrule
\end{tabular}
\end{table*}

\itpara{Workflow Notation in Fig.~\ref{fig:representative}.}
Labels $K_i$, $U_i$, and $V_i$ are mnemonic transaction names, with the current transaction circled.
Unmodalized fields refer to that transaction;
$\s{out}_i/\oum_i$ are its ledger-derived metadata (\S\ref{subsec:ledgermodel}).
The modalities $\inp_j/\oup_j$ refer to the $j$th predecessor and the transaction realizing successor position~$j$;
an unrealized successor yields unknown.
The marker $\lambda_i$ carries RI continuation state, while $\s{bal}$ is only a diagrammatic flow tag.
When a workflow joins branches or inputs, the checker-side RI-instance condition from \S\ref{subsec:algo} prevents combining distinct instances;
this constraint is absent from the label-abstracted semantics.



\begin{figure}[htbp]
	\centering

	\begin{minipage}{0.95\textwidth}

	\centering

	\hspace{-2em}
	\raisebox{.18\height}{
	\begin{subfigure}{0.47\textwidth}
		\centering
		~~\xymatrixcolsep{2.0pc}\xymatrix{
			& \text{\mycircled{$K_0$}}_{\lambda_1 = 0} \ar[r]^{\text{stake \& guess}}_{\s{bal} = 0} \ar[d]^{\text{stake \& guess}}_{\s{bal} = 0} & \text{\mycircled{$U_0$}}_{\lambda_1 = 1} \ar[d]^{\text{finalize}}_{\s{bal} = 1} & \ar[l]_{\ \ \ \s{bal} = 1} \\
			\ar[r]_{\s{bal} = 1} & \text{\mycircled{$U_1$}}_{\lambda_1 = 1} \ar[r]^{\text{finalize}}_{\s{bal} = 1} & \text{\mycircled{$K_1$}}_{\lambda_1 = 2} \ar[r]_{\ \ \ \s{bal} = 2} &
		}
		\caption{Guessing Game with an External Beacon}
		\label{subfig:game}
	\end{subfigure}
	}
	\hspace{-0.2em}
	\begin{subfigure}{0.47\textwidth}
		\centering
		~\xymatrixcolsep{0.4pc}\xymatrix{
			&K_0 \ar[d]^{\text{genesis}\ \ \ \ } & &
			\text{\mycircled{$V_2$}} \ar[r]^{\text{transfer}} &
			\text{\mycircled{$U_2$}} \ar@/_0.75pc/[dr]^{\text{merge}} &
			\\
			&
			\text{\mycircled{$U_1$}} \ar[r]^{\text{transfer}} &
			\text{\mycircled{$V_1$}} \ar@/_0.75pc/[ur]^{\text{split}} \ar@/^0.75pc/[dr]^{\text{split}} &
			& &
			\text{\mycircled{$V_4$}}
			\\
			 & & &
			\text{\mycircled{$V_3$}} \ar[r]^{\text{transfer}} &
			\text{\mycircled{$U_3$}} \ar@/^0.75pc/[ur]^{\text{merge}} &
		}
		\caption{Layer-Two Assets}
		\label{subfig:laytwo}
	\end{subfigure}


	\begin{subfigure}{0.75\textwidth}
		\begin{center}
		~~\xymatrixcolsep{2pc}\xymatrix{
			\text{\mycircled{$K_0$}} \ar[r]^{\text{out} 1}_{\textsf{bal} = 3} \ar@/_0.75pc/[rd]^{\text{out} 2}_{\textsf{bal} = 1} &
			\text{\mycircled{$K_1$}} \ar[r]^{\text{out} 1}_{\textsf{bal} = 2} \ar@/_0.75pc/[rd]^{\text{out} 2}_{\textsf{bal} = 1} &
			\text{\mycircled{$K_2$}} \ar[r]^{\text{out} 1}_{\textsf{bal} = 1} \ar@/_0.75pc/[rd]^{\text{out} 2}_{\textsf{bal} = 1} &
			\text{\mycircled{$K_3$}} \ar[r]^{\text{out} 1}_{\textsf{bal} = 0} \ar@/_0.75pc/[rd]^{\text{out} 2}_{\textsf{bal} = 1} &
			\circledast
			\\
			& U_1 & U_2 & U_3 & U_4 \\
		}
		\end{center}
		\caption{Collective Payments}
		\label{subfig:collp}
	\end{subfigure}


	\begin{subfigure}{0.45\textwidth}
		\centering
		~~\xymatrixcolsep{0.8pc}\xymatrix{
			& U_1 & U_2 & U_3
			\\
			\text{\mycircled{$K_0$}} \ar@/^0.75pc/[ur]_{\text{out} 1} \ar[r]_{\text{out} 2} \ar@/_0.75pc/[dr]_{\text{out} 3} &
			\text{\mycircled{$K_1$}} \ar@/^0.75pc/[ur]_{\text{out} 1} \ar[r]_{\text{out} 2} \ar@/_0.75pc/[dr]_{\text{out} 3} &
			\text{\mycircled{$K_2$}} \ar@/^0.75pc/[ur]_{\text{out} 1} \ar[r]_{\text{out} 2} \ar@/_0.75pc/[dr]_{\text{out} 3} & \cdots
			\\
			& \vdots & \vdots & \vdots
		}
		\caption{Taxed Transfers}
		\label{subfig:tax}
	\end{subfigure}
	\hfill
	\begin{subfigure}{0.45\textwidth}
		\centering
		~~\xymatrixcolsep{0.8pc}\xymatrix{
			& &
			\text{\mycircled{$K_4$}} \ar@{--}[]+<2.5em, 0em>;[dd]+<2.5em, 0em> &
			\\
			&
			\text{\mycircled{$K_2$}} \ar[r] \ar@/_0.75pc/[ur]&
			\text{\mycircled{$K_5$}} \ar[r]&
			K_7
			\\
			\text{\mycircled{$K_1$}} \ar[r] \ar@/_0.75pc/[ur] &
			\text{\mycircled{$K_3$}} \ar[r] &
			\text{\mycircled{$K_6$}} \ar[r] & K_8
			\\
		}
		\caption{Branch-Local Cooldown Flow}
		\label{subfig:temporal}
	\end{subfigure}
	\end{minipage}
\CaptionAndDescription{Case-study workflows:
We show only transaction dependencies and referenced successor positions, omitting the full RI formulas.
Circled nodes carry nontrivial RIs and propagate them along the indicated edges.
Here, $\lambda_1$ is a continuation marker, labels $\text{out}\ j$ denote successor positions, $\s{bal}$ is an informal flow annotation, and $K_i$, $U_i$, and~$V_i$ are diagrammatic transaction names; the dashed line in (e) marks expiry.%
}
	\label{fig:representative}
\end{figure}

\itpara{Non-Fungible Tokens (NFTs).}
The spending transaction uses $x_1$ as a mode flag ($x_1 = 1$ for transfer and $x_1 \neq 1$ for minting) and $x_2$ as the token identifier; the setup transaction has no predecessor carrying this RI.
The group fields $y_1$ and $y_2$ hold the current owner key and next-owner field (or witness material bound to it), respectively.
The signature $\sigma_1$ authorizes spending under $y_1$, while $\sigma_2$ authorizes minting under the fixed issuer key $\texttt{issuer} \in \mathbb{G}$.
Using successor position~$1$ for token continuation and $\oup_2 \bmbot$ to forbid a second successor position, the RI is:
\begin{align*}
F \coloneqq\;&
\oup_1 \s{RI}
\wedge
\oup_2 \bmbot
\wedge
\s{Ver}(y_1, \sigma_1)
\wedge
\Bigl(
\inp_1 \s{RI}
\to
\oup_1 x_1 = 1
\wedge
\oup_1 x_2 = x_2
\wedge
\oup_1 y_1 = y_2
\\
&\hspace{1.5em}
\wedge
\bigl(
\neg \inp_1 \inp_1 \s{RI}
\to
x_1 \ne 1
\bigr)
\wedge
\bigl(
x_1 \ne 1
\to
\s{SVer}\!\left(\texttt{issuer}, \s{H}(x_1, x_2, y_1, y_2, \Lambda), \sigma_2\right)
\wedge
\neg \inp_1 \inp_1 \s{RI}
\bigr)
\Bigr).
\end{align*}

If successor position~$1$ is realized, transfer ($x_1=1$) preserves the token identifier and updates the owner.
The condition $\neg \inp_1 \inp_1 \s{RI} \to x_1 \ne 1$ forces the post-setup transaction into mint mode, while the mint branch requires issuer authorization and $\neg \inp_1 \inp_1 \s{RI}$, preventing later reminting.
The final $\s{Ver}(y_1, \sigma_1)$ authenticates the spend, and the unconditional matching atom $\oup_1\s{RI}$ propagates the same RI formula along the setup--mint--transfer chain.
Binding mint authorization to $\Lambda$ limits cross-context replay;
global anti-duplication over an accepted height window additionally requires a corresponding recipient policy.
Checker-side RI-instance labels separate concurrent formula instances;
provenance and global uniqueness additionally depend on the RI formula and application policy.
Conventional UTXO realizations likewise thread ownership state through successive outputs;
\ctvp changes how the transition rule is carried and re-checked.

\itpara{Gaming with Guesses and an External Beacon.}
Along a propagated instance, the continuation marker $\lambda_1 \in \{0,1,2\}$ serves as the stage parameter for setup, guessing, and finalization, respectively.
In a guessing transaction, $x_1$ stores the player's guess, $y_1$ is the signer public key, and $\sigma_1$ authorizes the spend under $y_1$.
The exogenous value $\s{Rnd}$ supplies validation-time randomness only at finalization.
At setup, transactions realizing successor positions~$1$ and~$2$ must carry the same RI formula with continuation-marker value~$1$, while $\oup_3 \bmbot$ forbids a third successor position.
A guessing transaction requires any transaction realizing successor position~$1$ to carry the same RI formula with continuation-marker value~$2$, while $\oup_2 \bmbot$ forbids a second successor position.
Finalization consumes the two guessing branches, checks $\s{CorrWin}$, and terminates RI propagation.
The RI formula is:
\begin{align*}
F \coloneqq {}\ &
\bigl(
\lambda_1 = 0
\to
(
\oup_1 \s{RI}(1)
\wedge
\oup_2 \s{RI}(1)
\wedge
\s{Ver}(y_1, \sigma_1)
\wedge
\oup_3 \bmbot
)
\bigr)
\wedge
\\
&
\bigl(
\lambda_1 = 1
\to
(
\oup_1 \s{RI}(2)
\wedge
\s{Ver}(y_1, \sigma_1)
\wedge
\oup_2 \bmbot
)
\bigr)
\wedge
\\
&
\bigl(
\lambda_1 = 2
\to
\s{Ver}(y_1, \sigma_1)
\wedge \s{HasIn}(2)
\wedge \inp_1 \s{RI}(1)
\wedge \inp_2 \s{RI}(1)
\\
&\hspace{2em}{}
\wedge \inp_3 \bmbot
\wedge \oup_2 \bmbot
\wedge \s{CorrWin}
\wedge \oup_1 \s{Ver}(y_1, \sigma_1)
\wedge \oum_1 = 2
\bigr).
\end{align*}

Finalization requires two guessing inputs, forbids extra workflow inputs and successor position~$2$, checks $\s{CorrWin}$, and, if position~$1$ is realized, fixes its amount to~$2$ and authenticates the winner-controlled successor.
The winning condition is:
\begin{align*}
\s{CorrWin} \coloneqq {}\ &
\Bigl(
|\s{Rnd} - \inp_1 x_1| < |\s{Rnd} - \inp_2 x_1|
\to
\oup_1 y_1 = \inp_1 y_1
\Bigr)
\wedge
\\
&
\Bigl(
|\s{Rnd} - \inp_2 x_1| \le |\s{Rnd} - \inp_1 x_1|
\to
\oup_1 y_1 = \inp_2 y_1
\Bigr).
\end{align*}

$\s{Rnd}$ is externally supplied: the RI enforces consistency, while unpredictability and bias resistance depend on the external randomness source.

\itpara{Fungible Layer-Two Assets.}
Our model supports a basic layer-two asset where all coins originate from a single issuer that proposes the genesis transaction.
For each transaction, we use:
\begin{enumerate}[topsep = 0pt]
\item
$x_1$ as a mode, where $0$, $1$, $2$, and $3$ denote transfer, split, merge, and genesis, respectively.
\item
$x_2$ as the local asset amount, and a constant address \texttt{CONST} as the issuer.
\item
$y_1$ as the authorization key---the issuer at genesis and current owner otherwise---and $y_2$ as the next-owner key, \ie, the initial recipient at genesis and transfer recipient otherwise.
\end{enumerate}
The design is illustrated in Fig.~\ref{subfig:laytwo}.
In our DSL, it is captured by:
\begin{align*}
F \coloneqq{}\;& (0 \le x_1) \wedge (x_1 \le 3)
\\
&{}\bigwedge
\big(
x_1 = 0
\to
\inp_1 \s{RI}
\wedge
\inp_2 \bmbot
\wedge
\oup_2 \bmbot
\wedge
\left(
\oup_1 x_1 \ne 2
\to
\oup_1 x_2 = x_2
\right)
\\
&\hspace{3em}
\wedge
\oup_1 \s{RI}
\wedge
\left(
\oup_1 y_1 = y_2
\vee
\oup_1 x_1 = 2
\right)
\wedge
\s{Ver}(y_1, \sigma_1)
\big)
\\
&{}\bigwedge
\big(
x_1 = 1
\to
\inp_1 \s{RI}
\wedge
\inp_2 \bmbot
\wedge
\oup_3 \bmbot
\wedge
\oup_1 x_2 + \oup_2 x_2 = x_2
\wedge
\oup_1 x_1 \ne 2
\wedge
\oup_2 x_1 \ne 2
\\
&\hspace{3em}
\wedge
\oup_1 \s{RI}
\wedge
\oup_2 \s{RI}
\wedge
\s{Ver}(y_1, \sigma_1)
\wedge
\oup_1 y_1 = y_1
\wedge
\oup_2 y_1 = y_1
\big)
\\
\phantom{F \coloneqq{}}\;&{}\bigwedge
\bigl(
x_1 = 2
\to
\s{HasIn}(2)
\wedge
\inp_1 \s{RI}
\wedge
\inp_2 \s{RI}
\\
&\hspace{3em}
\wedge
\inp_3 \bmbot
\wedge
\oup_2 \bmbot
\wedge
\inp_1 x_2 + \inp_2 x_2 = x_2
\wedge
\inp_1 x_1 \ne 1
\wedge
\inp_2 x_1 \ne 1
\\
&\hspace{3em}
\wedge
\bigl(
\oup_1 x_1 \ne 2
\to
\oup_1 x_2 = x_2
\bigr)
\wedge
\oup_1 \s{RI}
\wedge
\bigl(
\oup_1 y_1 = y_1
\vee
\oup_1 x_1 = 2
\bigr)
\\
&\hspace{3em}
\wedge
\s{Ver}(y_1, \sigma_1)
\wedge
\inp_1 y_1 = y_1
\wedge
\inp_2 y_1 = y_1
\bigr)
\\
&{}\bigwedge
\big(
x_1 = 3
\to
y_1 = \texttt{CONST}
\wedge
\s{Ver}(y_1, \sigma_1)
\wedge \oup_1 x_2 = x_2
\wedge \oup_1 y_1 = y_2
\wedge
\oup_1 \s{RI}
\big).
\end{align*}
Transfer and split require one RI-bearing input;
merge requires two, with asset amounts summing to the current amount and owner $y_1$.
Genesis is issuer-authorized and, if successor position~$1$ is realized, carries the initial amount and recipient into the RI lineage.
The guard requires $x_1 \in \{0, 1, 2, 3\}$.

\itpara{Collective Payments.}
A sender may pay recipients through a claim chain rather than one enumerating all payees.
We use an algebraic membership-witness chain only to exercise RI state evolution;
no accumulator-security claim is made.
\begin{enumerate}[topsep = 0pt]
\item
The sender fixes the membership instance.
For a recipient address $y_1$, a witness $x_1$ satisfies
$
y_2^{\s{hashToInt}(y_1) \cdot x_1} = y_3
$,
where $y_2$ is a generator and $y_3$ is the current membership-state value.
\item
A recipient claims by providing $x_1$ and realizing the designated continuation position.
The constraint $\oup_1 \inp_2 \bmbot$ enforces the chain shape illustrated in Fig.~\ref{subfig:collp}:
the sender starts the payments at $K_0$, followed by claims at $K_1, K_2, \ldots$
\end{enumerate}
\begin{align*}
F \coloneqq {}\ &
\bigl(y_2^{\s{hashToInt}(y_1) \cdot x_1} = y_3\bigr)
\wedge
(\oup_1 y_2 = y_2)
\wedge
(\oup_1 y_3 = y_2^{x_1})
\wedge
\s{Ver}(y_1, \sigma_1) \ \wedge \\
&\oup_1 \s{RI}(\lambda_1 + 1)
\wedge
\oup_1 \inp_2 \bmbot
\wedge
\oup_2 y_1 = y_1
\wedge
\oup_2 \s{Ver}(y_1, \sigma_1)
\wedge
\oum_2 = 1.
\end{align*}
If successor position~$2$ is realized, it is a one-unit payout bound to the claimant key and authentication.
The continuation amount and positions beyond~$2$ remain unconstrained.

\itpara{Taxed Transfers.}
A transfer-tax-style policy can relate a transfer amount to a designated fee amount.
Let $\texttt{RECEIVER} \in \mathbb{G}$ be the designated tax-receiver key.
For the workflow in Fig.~\ref{fig:representative}d, we use
\[
F \coloneqq
\inp_2 \bmbot
\wedge \oup_4 \bmbot
\wedge \oup_2 \s{RI}
\wedge 9 \cdot \oum_3 = \oum_1
\wedge \s{Ver}(y_1, \sigma_1)
\wedge \oup_3 y_1 = \texttt{RECEIVER}
\wedge \oup_3 \s{Ver}(y_1, \sigma_1).
\]
If positions~$1$ and~$3$ are realized, the RI enforces $\oum_3=\oum_1/9$ and authenticates the tax-key successor.
Realization of those positions, binding $\oum_1$ to total input, and transfer-ownership preservation require additional clauses.

\itpara{Branch-Local Cooldown Flow.}
Let $\s{minGap}$ be the required inclusion-height gap and $\s{exTime}$ the expiration height.
We model a per-branch cooldown requiring the parent-to-child inclusion gap to exceed $\s{minGap}$ until an expiration point, as illustrated in Fig.~\ref{subfig:temporal}.
The field $x_1$ is a signer-chosen not-before height for the current transaction, while $\Lambda$ is its fixed inclusion height:
\[
\oup_3 \bmbot
\wedge
\inp_2 \bmbot
\wedge
\bigl(
\Lambda < \s{exTime}
\rightarrow
\bigl(
x_1 - \inp_1 \Lambda > \s{minGap}
\wedge
x_1 \le \Lambda
\wedge
\oup_1 \s{RI}
\wedge
\oup_2 \s{RI}
\bigr)
\bigr)
\wedge
\s{Ver}(y_1, \sigma_1).
\]
Before expiry, $x_1 \le \Lambda$ enforces the not-before height and
$x_1 - \inp_1 \Lambda > \s{minGap}$ the gap threshold;
positions~$1$ and~$2$ propagate the RI.
Because they may branch, the enforced gap is branch-local;
successor ownership also requires key binding.
Further cases, including commit-and-reveal gaming and collective receiving, appear in
\ifnum\fullFlag = 1
Appendix~\ref{app:usecase}.
\else
Appendix~B.1~\cite{eprint/TangCFGL26}.
\fi

\subsection{Implementation and Evaluation Methodology}
\label{subsec:comparison}

\itpara{Benchmark Interpreter.}
We implement a Python~3.13.2 RI interpreter and cost estimator for the six reported benchmark traces.
It evaluates the benchmark RI encodings on fixed traces and reports the validation-cost proxy used below.

\itpara{Benchmark Snapshot.}
The archived benchmark snapshot corresponds to the accepted-version evaluation and is retained for reproducibility.
Some final camera-ready specification refinements are not reflected in those benchmark encodings;
the formal definitions in this paper are authoritative for the final DSL and semantics.


\begin{figure}[!ht]
	\centering

	\begin{subfigure}{0.45 \textwidth}
	\includegraphics[width = \linewidth]{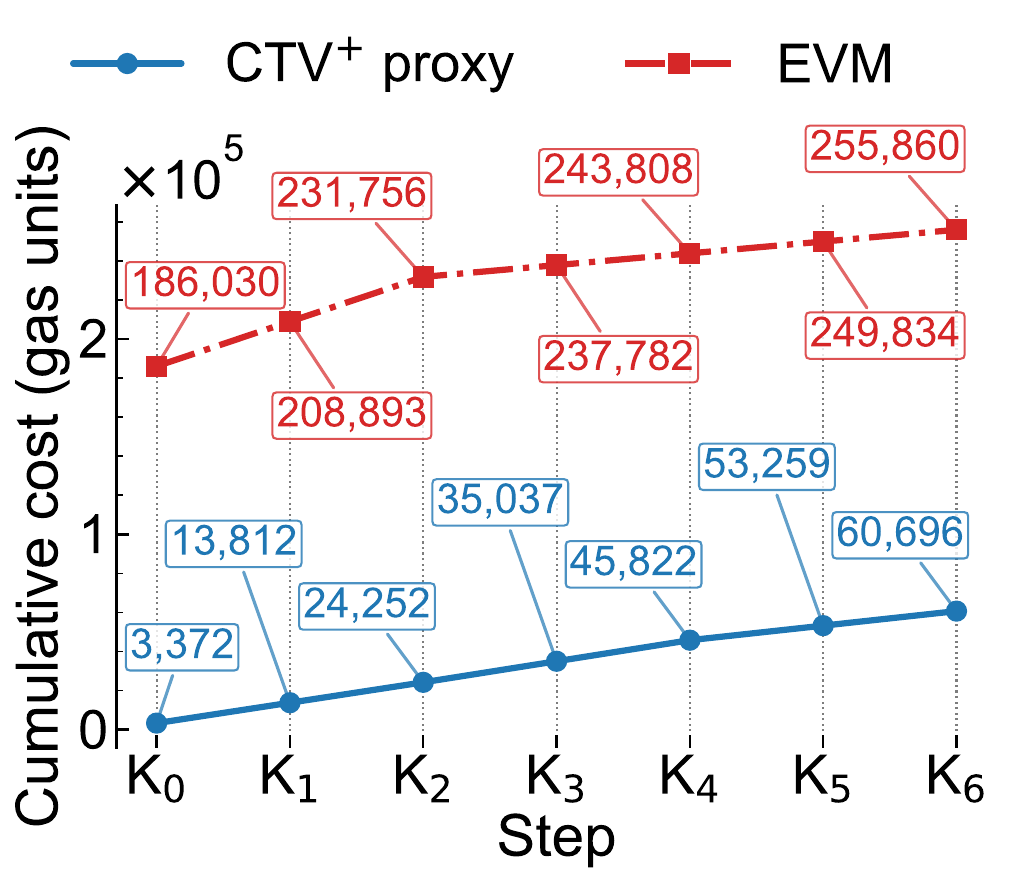}
	\caption{Non-Fungible Tokens}
	\label{evalfig:nft}
	\end{subfigure}
	\begin{subfigure}{0.45 \textwidth}
	\includegraphics[width = \linewidth]{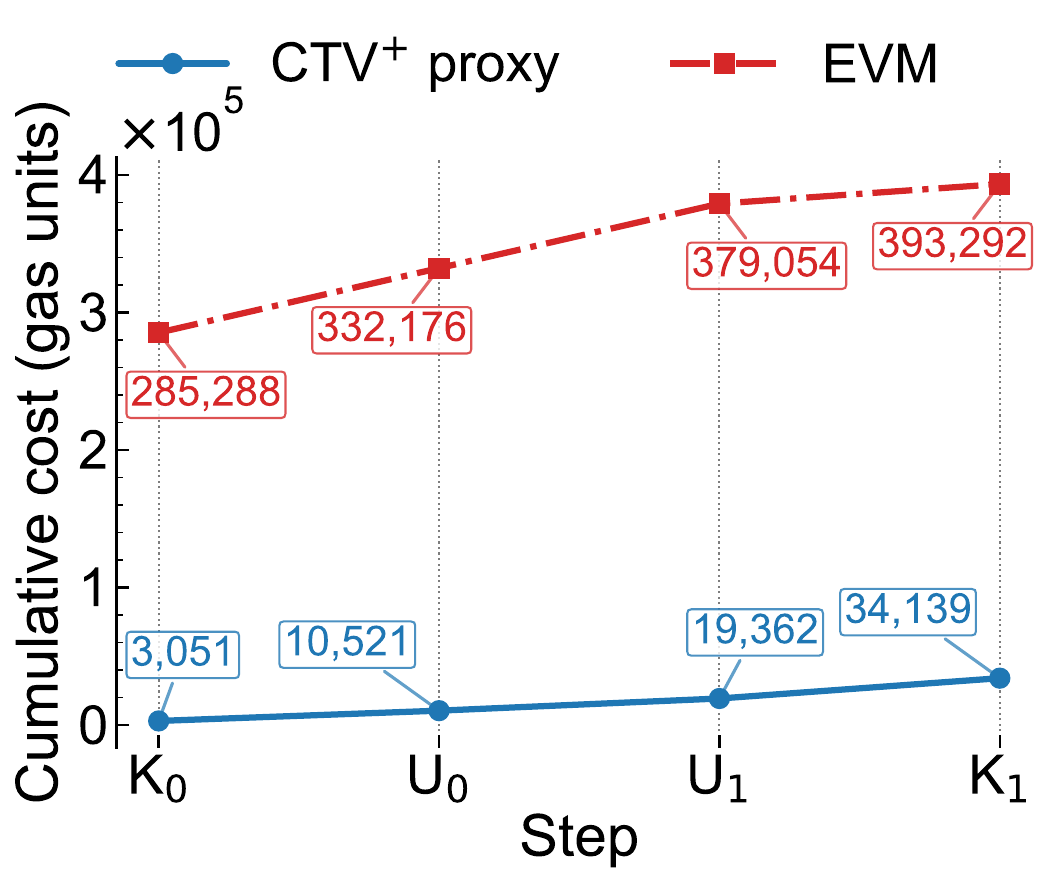}
	\caption{Gaming with Guesses and an External Beacon}
	\label{evalfig:game}
	\end{subfigure}

	\begin{subfigure}{0.45 \textwidth}
	\includegraphics[width = \linewidth]{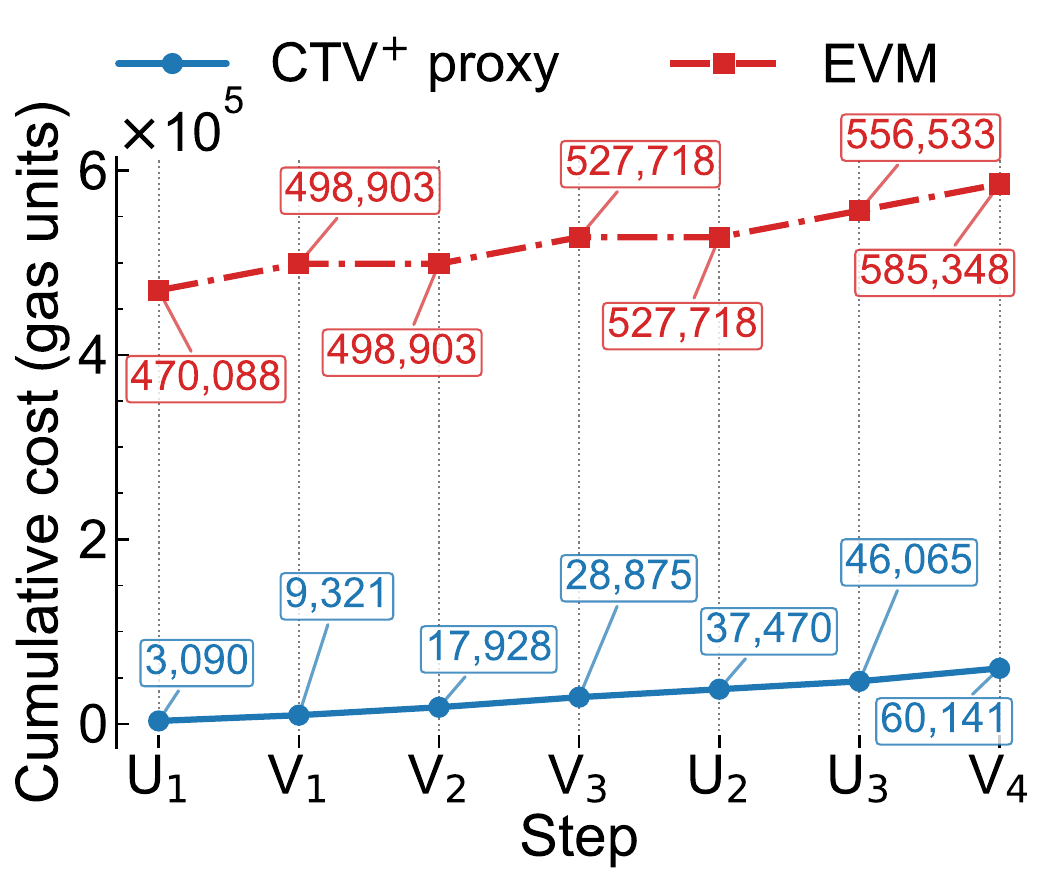}
	\caption{Fungible Layer-Two Assets}
	\label{evalfig:laytwo}
	\end{subfigure}
	\begin{subfigure}{0.45 \textwidth}
	\includegraphics[width = \linewidth]{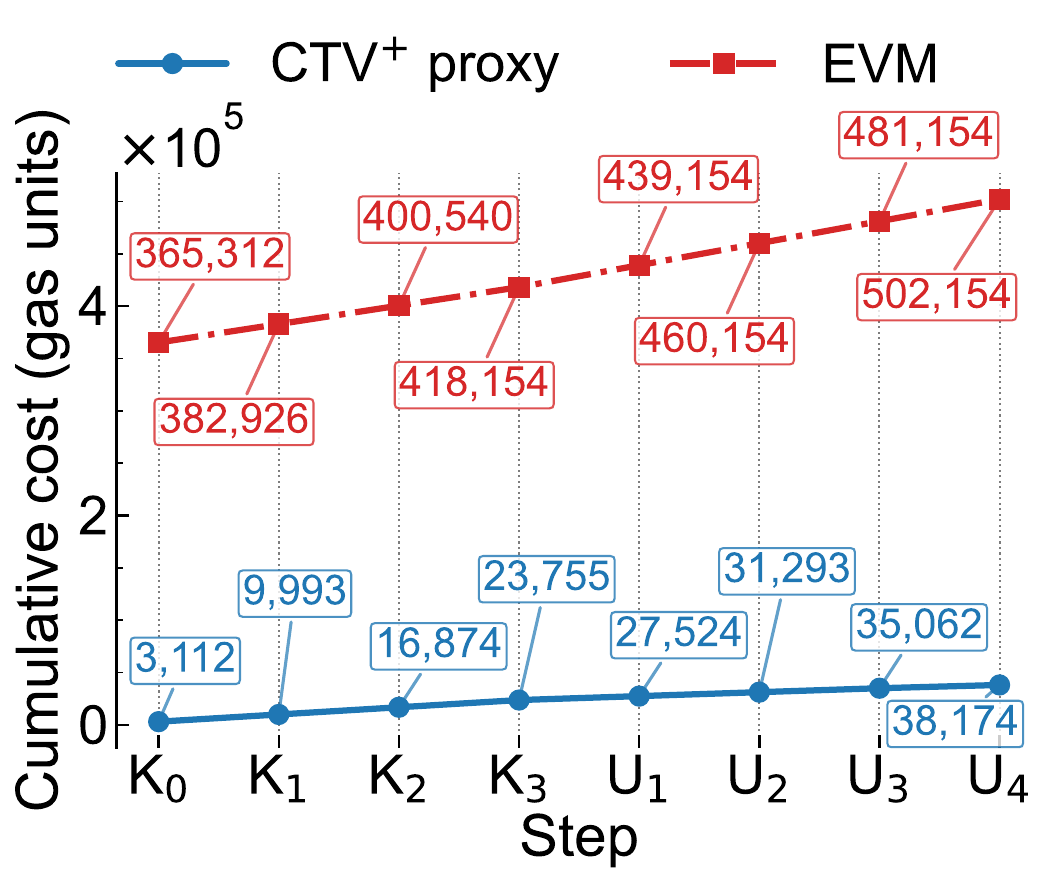}
	\caption{Collective Payments}
	\label{evalfig:collpay}
	\end{subfigure}

	\begin{subfigure}{0.45 \textwidth}
	\includegraphics[width = \linewidth]{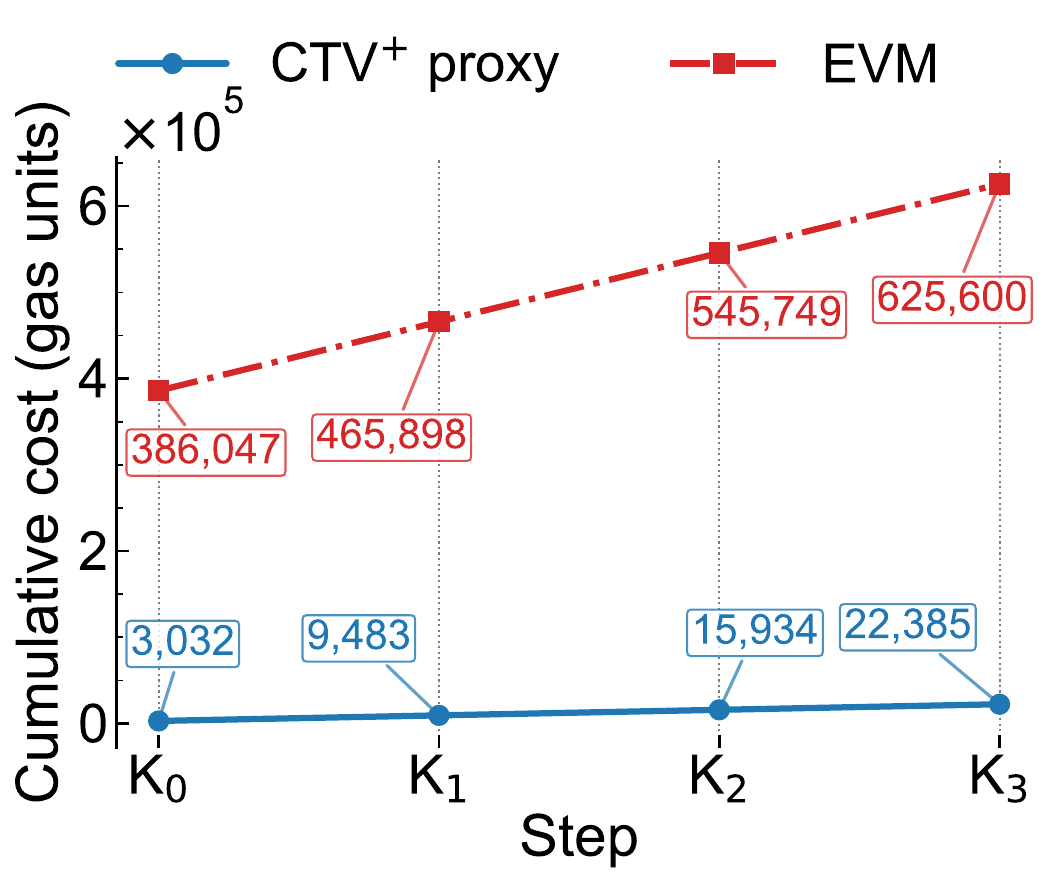}
	\caption{Taxed Transfers}
	\label{evalfig:tax}
	\end{subfigure}
	\begin{subfigure}{0.45 \textwidth}
	\includegraphics[width = \linewidth]{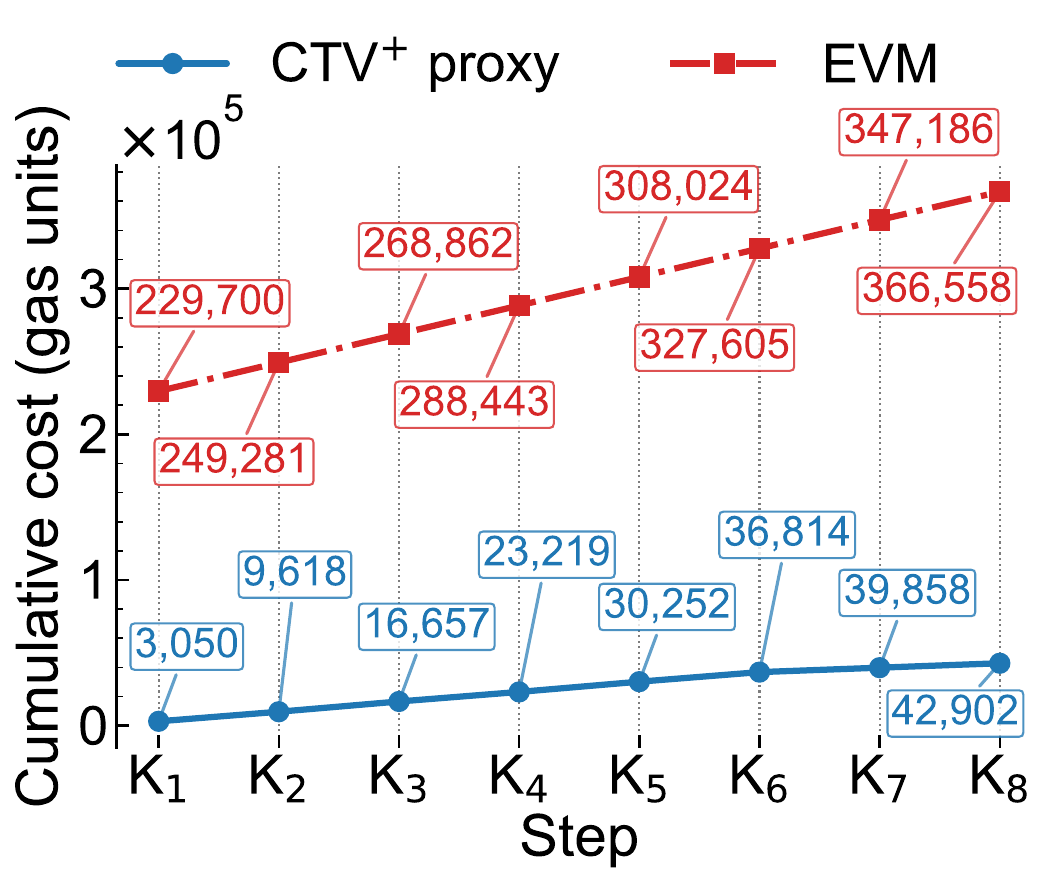}
	\caption{Branch-Local Cooldown Flow}
	\label{evalfig:timelimited}
	\end{subfigure}
	\CaptionAndDescription{Cumulative \ctvp proxy and measured EVM gas:
	(a)--(e) trace-aligned; (f) linear reference.}
	\label{fig:evalRes}
\end{figure}

\itpara{Comparison Methodology and Baseline Scope.}
For five case studies, each realized RI transition is paired with Solidity calls for the corresponding step-admission checks on the measured trace;
this is trace alignment, not all-trace semantic equivalence.
The branch-local cooldown uses a simpler linear reference.
The benchmark interpreter reports a validation-cost proxy for current- and parent-RI checks, normalized by Ethereum Virtual Machine (EVM) opcode gas weights;\footnote{See \url{https://ethereum.org/developers/docs/evm/opcodes}.}
the proxy is not EVM gas, while Solidity values are measured gas and include transaction overhead.

\itpara{Evaluation Environment.}
We implement the Solidity baselines in Solidity~0.8.27 and run them in Remix VM (Prague) on an Apple M3 with $16$~GB RAM and macOS~15.6.1 (24G90).
For the six workloads of \S\ref{subsec:usecase}, Table~\ref{tab:anchors} gives the practice anchors, while
\ifnum\fullFlag = 1
Appendix~\ref{app:deploymentanchors}
\else
Appendix~B.2 of the full version~\cite{eprint/TangCFGL26}
\fi
gives the baseline instantiations and mappings.

\itpara{Workload Matching.}
For the five trace-aligned cases, each benchmark step or grouped segment is paired with the EVM calls for the corresponding realized transition, summing gas when needed.
Integers use 256-bit words, and group elements use field elements modulo $p<2^{256}$, matching EVM calldata/storage representation.
\ifnum\fullFlag = 1
The input encoding, baseline scope, and step mapping appear in
Appendix~\ref{app:deploymentanchors}.
\fi
For the branch-local cooldown, we report a simpler linear eight-transaction Solidity reference because no trace-aligned branching reference was implemented.

\subsection{Results}
The cumulative \ctvp benchmark proxy increases approximately linearly over all six measured traces.
For the five trace-aligned benchmark cases, the reported endpoint proxy-to-gas ratios are $3.6\%$--$23.7\%$ (unweighted mean $10.8\%$).
The reported cooldown reference ratio is $11.7\%$ and is excluded from that summary because its Solidity workload is not trace-aligned.
These proxy-to-gas ratios do not measure runtime speedups or deployable fee reductions;
EVM deployment/state costs and \ctvp transaction-carried validation differ structurally.


\section{Conclusion}
Recursive invariants (RIs) support staged multi-step UTXO workflows by carrying persistent workflow rules across realized successor positions, without application-level shared mutable contract state.
The DSL uses Kleene's three-valued semantics to defer future-dependent clauses until the relevant successor is realized, and its deduction system is sound for the evaluable fragment.
Under mathematical-integer semantics and an unbounded transaction chain, this fragment can simulate UTXO counter machines.
Across six benchmark workloads, the cumulative validation-cost proxy grows approximately linearly over the measured traces.
RI re-checking reduces reliance on preconstructed transactions while preserving validation-time locality and explicit cost accounting.

\subsection*{Data-Availability Statement}
The Zenodo artifact archives the frozen benchmark code and recorded outputs, the Solidity comparison implementations, and an accompanying Coq development~\cite{zenodo/ctvplus-artifact}.
It also includes scripts for executing the archived Python and Foundry benchmark suites.
The latest version of the source code is available at \url{https://github.com/htftsy/CTVPlus}.

\subsection*{\acksname}
Sherman Chow is supported in part by the Research Grants Council of Hong Kong under the Collaborative Research Fund (C5097-25G) and the General Research Fund (14210825).
Guoqiang Li is supported by the National Natural Science Foundation of China with Grant No.~62572294.

\citestyle{acmnumeric}
\bibliographystyle{ACM-Reference-Format}
\bibliography{fc,misc,new-misc}

@inproceedings{fc/CromanDEGJKMSSS16,
  author       = {Kyle Croman and
                  Christian Decker and
                  Ittay Eyal and
                  Adem Efe Gencer and
                  Ari Juels and
                  Ahmed E. Kosba and
                  Andrew Miller and
                  Prateek Saxena and
                  Elaine Shi and
                  Emin G{\"{u}}n Sirer and
                  Dawn Song and
                  Roger Wattenhofer},
  title        = {On Scaling Decentralized Blockchains - {(A} Position Paper)},
  booktitle    = {FC},
  pages        = {106--125},
  year         = {2016},
doi          = {10.1007/978-3-662-53357-4\_8},
}

@inproceedings{fc/MoserES16,
  author       = {Malte M{\"{o}}ser and
                  Ittay Eyal and
                  Emin G{\"{u}}n Sirer},
  title        = {Bitcoin Covenants},
  booktitle    = {Bitcoin workshop co-located with FC},
  pages        = {126--141},
  year         = {2016},
 doi          = {10.1007/978-3-662-53357-4\_9},
}

@inproceedings{fc/OConnorP17,
  author       = {Russell O'Connor and
                  Marta Piekarska},
  title        = {Enhancing Bitcoin Transactions with Covenants},
  booktitle    = {FC},
  pages        = {191--198},
  year         = {2017},
  doi          = {10.1007/978-3-319-70278-0\_12},
}

@inproceedings{fcw/ChakravartyCMMJ20,
  author       = {Manuel M. T. Chakravarty and
                  James Chapman and
                  Kenneth MacKenzie and
                  Orestis Melkonian and
                  Michael Peyton Jones and
                  Philip Wadler},
  title        = {The Extended {UTXO} Model},
  booktitle    = {WTSC co-located with FC},
  pages        = {525--539},
  year         = {2020},
 doi          = {10.1007/978-3-030-54455-3\_37},
}

@inproceedings{ccs/BartolettiZ18,
  author       = {Massimo Bartoletti and
                  Roberto Zunino},
  title        = {{BitML}: {A} Calculus for Bitcoin Smart Contracts},
  booktitle    = {CCS},
  pages        = {83--100},
  year         = {2018},
 doi          = {10.1145/3243734.3243795},
}

@inproceedings{fmbc/MelkonianSC25,
  author       = {Orestis Melkonian and
                  Wouter Swierstra and
                  James Chapman},
  title        = {Program Logics for Ledgers},
  booktitle = {Formal Methods for Blockchains (FMBC), co-located with ETAPS},
  pages        = {10:1--10:22},
  year         = {2025},
  doi          = {10.4230/OASICS.FMBC.2025.10},
}

@inproceedings{esop/HajduJ20,
  author       = {{\'{A}}kos Hajdu and
                  Dejan Jovanovic},
  title        = {{SMT}-Friendly Formalization of the {Solidity} Memory Model},
  booktitle    = {ESOP},
  pages        = {224--250},
  year         = {2020},
 doi          = {10.1007/978-3-030-44914-8\_9},
}

@article{fgcs/BartolettiMZ26,
  author       = {Massimo Bartoletti and
                  Riccardo Marchesin and
                  Roberto Zunino},
  title        = {Scalable {UTXO} smart contracts via fine-grained distributed state},
  journal      = {Future Gener. Comput. Syst.},
  volume       = {175},
  pages        = {108023},
  year         = {2026},
 doi          = {10.1016/j.future.2025.108023},
}

@article{pacmpl/SergeyNJ0TH19,
  author       = {Ilya Sergey and
                  Vaivaswatha Nagaraj and
                  Jacob Johannsen and
                  Amrit Kumar and
                  Anton Trunov and
                  Ken Chan Guan Hao},
  title        = {Safer smart contract programming with {Scilla}},
  journal      = {Proc. ACM Program. Lang. (OOPSLA)},
volume = {3},
  pages        = {185:1--185:30},
  year         = {2019},
  doi = {10.1145/3360611},
}

@inproceedings{csf/BartolettiLZ21,
  author       = {Massimo Bartoletti and
                  Stefano Lande and
                  Roberto Zunino},
  title        = {Computationally sound {Bitcoin} tokens},
  booktitle    = {CSF},
  numpages        = {15},
  year         = {2021},
  doi = {10.1109/CSF51468.2021.00022},
}

@inproceedings{lsfa/VinogradovaS24,
  author    = {Polina Vinogradova and Alexey Sorokin},
  title     = {Properties of {UTxO} Ledgers and Programs Implemented on Them},
  booktitle    = {LSFA},
  pages        = {1--20},
  year         = {2024},
  month        = jun,
  doi          = {10.4204/EPTCS.421.1},
}

@article{mst/FischerMR68,
  author       = {Patrick C. Fischer and
                  Albert R. Meyer and
                  Arnold L. Rosenberg},
  title        = {Counter Machines and Counter Languages},
  journal      = {Math. Syst. Theory},
  volume       = {2},
  number       = {3},
  pages        = {265--283},
  year         = {1968},
  doi = {10.1007/BF01694011},
}

@inproceedings{nss/MengWWLYLZC18,
  author    = {Weizhi Meng and Jianfeng Wang and Xianmin Wang and Joseph K. Liu and Zuoxia Yu and Jin Li and Yongjun Zhao and Sherman S. M. Chow},
  title     = {Position Paper on Blockchain Technology: {Smart} Contract and Applications},
  booktitle = {NSS},
  pages     = {474--483},
  year      = {2018},
  doi       = {10.1007/978-3-030-02744-5\_35},
}

@book{book/Kleene52,
  author    = {Stephen Cole Kleene},
  title     = {Introduction to Metamathematics},
  publisher = {D. Van Nostrand Company},
  address   = {New York, NY, USA},
  year      = {1952}
}

@misc{bip119,
  author       = {Jeremy Rubin},
  title        = {{BIP} 119: {OP}\_{CHECKTEMPLATEVERIFY}},
  year         = {2020},
  howpublished = {\url{https://bips.dev/119}}
}

@misc{bitvm,
  author       = {Lukas Aumayr and
                  Zeta Avarikioti and
                  Robin Linus and
                  Matteo Maffei and
                  Andrea Pelosi and
                  Christos Stefo and
                  Alexei Zamyatin},
  title        = {{BitVM}: {Quasi-Turing} Complete Computation on {Bitcoin}},
  howpublished      = {{IACR} Cryptol. ePrint Arch. 2024/1995},
  year         = {2024}
}

@manual{tool/coq,
  title        = {The {Coq} Proof Assistant Reference Manual, Version 8.19.0},
  author       = {{Coq Development Team}},
  year         = {2024},
  howpublished = {\url{https://rocq-prover.org/doc/V8.19.0/refman/}},
  note         = {Accessed: 2026-09-12},
  organization = {Coq Development Team}
}

@techreport{tool/move,
  title        = {{Move: A} Language With Programmable Resources},
  author       = {Sam Blackshear and Evan Cheng and David L. Dill and Victor Gao and Ben Maurer and Todd Nowacki and Alistair Pott and Shaz Qadeer and Rain and Dario Russi and Stephane Sezer and Tim Zakian and Runtian Zhou},
  year         = {2019},
  institution  = {Libra},
  type         = {Technical Report},
  note         = {Accessed: 2026-09-12},
  howpublished = {\url{https://developers.libra.org/docs/assets/papers/libra-move-a-language-with-programmable-resources.pdf}}
}

@misc{tool/plutus,
  title        = {{Plutus}: The Smart Contract Language for Cardano},
  author       = {{Input Output Global (IOG)}},
  year         = {2025},
  howpublished = {\url{https://docs.cardano.org/developer-resources/smart-contracts/plutus}},
  note         = {Accessed: 2026-09-12}
}

@misc{eprint/TangCFGL26,
  author  = {Shuyang Tang and Sherman S. M. Chow and Hongfei Fu and Zihan Guo and Guoqiang Li},
  title   = {Staged Multi-step {UTXO} Workflows via Recursive Invariants},
  year    = {2026},
  note = {Full version}
}

@misc{zenodo/ctvplus-artifact,
  author    = {Shuyang Tang and Sherman S. M. Chow and Hongfei Fu and Zihan Guo and Guoqiang Li},
  title     = {Staged Multi-step {UTXO} Workflows via Recursive Invariants: Artifact},
  year      = {2026},
  publisher = {Zenodo},
  version   = {1.2},
  doi       = {10.5281/zenodo.21702929},
}

@misc{BIP345_OPVAULT,
  author       = {James O'Beirne and Greg Sanders},
  title        = {{BIP 345: OP\_VAULT}},
  howpublished = {\url{https://bips.dev/345}},
  year         = {2023},
  month        = feb,
  note         = {Accessed: 2026-09-12}
}

@misc{Optech_Topic_Vaults,
  author       = {{Bitcoin Optech}},
  title        = {Vaults},
  year         = {2026},
  month        = jan,
  howpublished          = {\url{https://bitcoinops.org/en/topics/vaults}},
  note         = {Accessed: 2026-09-12}
}

@misc{Optech_Topic_CTV,
  author       = {{Bitcoin Optech}},
  title        = {{OP\_CHECKTEMPLATEVERIFY}},
  howpublished = {\url{https://bitcoinops.org/en/topics/op_checktemplateverify}},
  year         = {2026},
  month        = jan,
  note         = {Accessed: 2026-09-12}
}

@misc{LightningTaprootAssetsOverview24,
  author       = {{Lightning Labs}},
  title        = {Taproot Assets},
  howpublished = {\url{https://docs.lightning.engineering/the-lightning-network/taproot-assets}},
  year         = {2024},
  month        = mar,
  note         = {Accessed: 2026-09-12}
}

@misc{LightningTaprootAssetsProtocol24,
  author       = {{Lightning Labs}},
  title        = {Taproot Assets Protocol},
  howpublished = {\url{https://docs.lightning.engineering/the-lightning-network/taproot-assets/taproot-assets-protocol}},
  year         = {2024},
  month        = sep,
  note         = {Accessed: 2026-09-12}
}

@misc{RGBClientSideValidation24,
  author       = {{RGB Docs}},
  title        = {Client-side Validation},
  howpublished = {\url{https://docs.rgb.info/distributed-computing-concepts/client-side-validation}},
  year         = {2024},
  month        = aug,
  note         = {Accessed: 2026-09-12}
}

@misc{RGBSingleUseSeals24,
  author       = {{RGB Docs}},
  title        = {Single-use Seals and Proof of Publication},
  howpublished = {\url{https://docs.rgb.info/distributed-computing-concepts/single-use-seals}},
  year         = {2024},
  month        = aug,
  note         = {Accessed: 2026-09-12}
}

@misc{RGBStateTransitions25,
  author       = {{RGB Docs}},
  title        = {Contract Operations: State Transitions},
  howpublished = {\url{https://docs.rgb.info/rgb-state-and-operations/state-transitions}},
  year         = {2025},
  month        = jul,
  note         = {Accessed: 2026-09-12}
}

@misc{SolidityGlobals,
  year    = {2026},
  author       = {{Solidity Project}},
  title        = {Units and Globally Available Variables},
  howpublished = {\url{https://docs.soliditylang.org/en/latest/units-and-global-variables.html}},
  note         = {Accessed: 2026-09-12},
}

@misc{EIP4399,
  author       = {Mikhail Kalinin and Dankrad Feist and Alexey Akhunov},
  title        = {{EIP-4399: S}upplant {DIFFICULTY O}pcode with {PREVRANDAO}},
  howpublished = {\url{https://eips.ethereum.org/EIPS/eip-4399}},
  year         = {2022},
  note         = {Accessed: 2026-09-12},
}

@misc{ChainlinkVRF,
year = {2024},
  author       = {{Chainlink}},
  title        = {{VRF} (Verifiable Random Function) Documentation},
  howpublished = {\url{https://docs.chain.link/vrf}},
  note         = {Accessed: 2026-09-12},
}

@misc{UniswapCommonErrors,
year = {2026},
  author       = {{Uniswap}},
  title        = {Uniswap Troubleshooting (v2)},
  howpublished = {\url{https://developers.uniswap.org/docs/protocols/v2/guides/troubleshooting}},
  note         = {Accessed: 2026-09-12},
}

@misc{UniswapRouter02,
  year         = {2020},
  author       = {{Uniswap}},
  title        = {{UniswapV2Router02.sol}},
  howpublished = {\url{https://github.com/Uniswap/v2-periphery/blob/master/contracts/UniswapV2Router02.sol}},
  note         = {Accessed: 2026-09-12},
}

@misc{UniswapFoTIssue,
  author       = {{Uniswap Interface Developers}},
  title = {Support of fee-on-transfer tokens},
  year = {2020},
  howpublished = {\url{https://github.com/Uniswap/interface/issues/835}},
  note = {GitHub issue \#835. Accessed: 2026-09-12}
}

@misc{EIP721,
  author       = {William Entriken and Dieter Shirley and Jacob Evans and Nastassia Sachs},
  title        = {{EIP-721: N}on-Fungible Token Standard},
  howpublished = {\url{https://eips.ethereum.org/EIPS/eip-721}},
  year         = {2018},
  note         = {Accessed: 2026-09-12},
}

@misc{OZContracts,
year = {2024},
  author       = {{OpenZeppelin}},
  title        = {{ERC721}: {OpenZeppelin} Contracts v4.x {API} Reference},
  howpublished = {\url{https://docs.openzeppelin.com/contracts/4.x/api/token/erc721}},
  note         = {Accessed: 2026-09-12},
}

@misc{SafeModules,
  author       = {{Safe}},
  title        = {Safe Modules (smart account modules)},
  howpublished = {\url{https://docs.safe.global/advanced/smart-account-modules}},
  year         = {2025},
  note         = {Accessed: 2026-09-12}
}

@misc{SafeSpendingLimits,
  author       = {{Safe}},
  title        = {Set up and use Spending Limits},
  howpublished = {\url{https://help.safe.global/articles/3961440620-set-up-and-use-spending-limits}},
  year         = {2026},
  note         = {Accessed: 2026-09-12}
}

@misc{SafeAllowanceModule,
  author       = {{Safe Foundation}},
  title        = {Allowance Module},
  howpublished = {\url{https://github.com/safe-fndn/safe-modules/tree/main/modules/allowances}},
  year = {2026},
  month = {Mar},
  note         = {Accessed: 2026-09-12},
}

@misc{UniswapMerkleDistributor,
  author  = {{Uniswap Labs}},
  title   = {Merkle Distributor},
  year    = {2020},
  url     = {https://github.com/Uniswap/merkle-distributor},
  urldate = {2026-09-12},
}

\ifnum\fullFlag = 1
\newpage
\appendix

\section{Semantic Details Omitted from the Main Text}
\label{app:semantics}

\subsection{Exogenous Symbols and Validation-Time Instantiation}
\label{app:exogenous}

\itpara{Compilation of Exogenous Symbols.}
We treat $\Lambda$ and $\s{Rnd}$ as exogenous parameters determined by a transaction's canonical inclusion context.
For an accepted transaction $\tx$, let $\s{height}_{\Gamma}(\tx)$ and $\s{rnd}_{\Gamma}(\tx)$ denote, respectively, the block height and the consensus-determined randomness value contained in or deterministically derived from that context.
When $\tx$ is later re-checked under a ledger extension, its unmodalized $\Lambda$ and $\s{Rnd}$ continue to evaluate to these fixed values rather than to the inclusion context of the later extension.
A modal occurrence evaluates the symbol in the inclusion context of the referenced transaction.
Thus, ledger extension can resolve a previously unknown successor-dependent clause without changing exogenous values already fixed for an earlier transaction.
These conventions suffice for our case studies and do not change the decidability of one-step validation.
A general outer logic for reasoning about exogenous sources across traces
is outside the present semantics.

\subsection{Evaluation of Arithmetic, Group, and Boolean Expressions}
\label{app:eval}

This appendix records the full clause-by-clause definitions of $\aeval$, $\geval$, and $\beval$.
The modal and RI-propagation clauses are the nonstandard part and are repeated here for completeness.
Following \S\ref{subsec:dsl}, hash arguments are canonically serialized with type tags.
Let $\s{encEval}_{\Gamma, \tx}(e)$ denote this tagged serialization after evaluating an arithmetic argument with $\aeval$ or a group argument with $\geval$.
If any argument evaluates to $\Omega$, the enclosing hash expression evaluates to $\Omega$ by strict propagation.
We use $H_a \colon \{0, 1\}^* \to \mathbb{Z}$ and $H_g \colon \{0, 1\}^* \to \mathbb{G}$ for arithmetic- and group-valued hash outputs, respectively, and write $\concat$ for concatenation.

Throughout this appendix, $\s{prev}_i(\tx_1, \tx_2)$ abbreviates $\tx_1.\s{in}_i = \tx_2.\id$.
The ledger-relative relation $\s{succ}^{\Gamma}_i(\tx_1, \tx_2)$ holds iff the ledger-derived metadata of $\tx_1$ in $\Gamma$ satisfies $\tx_1.\s{out}_i = \tx_2.\id$.

\begin{definition}[Arithmetic Expression Evaluation]
Evaluation $\aeval \colon aexpr^{\star} \to \mathbb{Z} \cup \{\Omega\}$
in the context of ledger $\Gamma$ and transaction $\tx$ is defined by structural recursion.
We use strict propagation of $\Omega$:
if any required immediate subexpression evaluates to $\Omega$, then the whole expression evaluates to~$\Omega$.
Otherwise:
\begin{align*}
\aeval(n)
&\coloneqq n
&& (n \in \mathbb{Z}),
\\
\aeval(\s{var})
&\coloneqq \tx.\s{var}
&& \bigl(\s{var} \in S_a' \setminus
	\{\s{out}_j, \oum_j \colon j \ge 1\}\bigr),
\\
\aeval(\Lambda)
&\coloneqq \s{height}_{\Gamma}(\tx),
&
\aeval(\s{Rnd})
&\coloneqq \s{rnd}_{\Gamma}(\tx),
\\
\aeval(\Omega)
&\coloneqq \Omega.
\end{align*}
\begin{align*}
\aeval(\s{out}_j)
&\coloneqq
\begin{cases}
\tx.\s{out}_j
& \text{if $\tx.\s{out}_j$ is defined in $\Gamma$,} \\
\Omega
& \text{otherwise,}
\end{cases}
\\
\aeval(\oum_j)
&\coloneqq
\begin{cases}
\tx.\oum_j
& \text{if $\tx.\oum_j$ is defined in $\Gamma$,} \\
\Omega
& \text{otherwise.}
\end{cases}
\end{align*}
\begin{align*}
\aeval(-a)
&\coloneqq -\,\aeval(a),
&
\aeval(a + b)
&\coloneqq \aeval(a) + \aeval(b),
\\
\aeval(a \cdot b)
&\coloneqq \aeval(a) \cdot \aeval(b),
&
\aeval(|a|)
&\coloneqq \left\lvert \aeval(a) \right\rvert.
\end{align*}
\[
\aeval\bigl(\s{H}(e_1, \ldots, e_\ell)\bigr)
\coloneqq
H_a\bigl(
\s{encEval}_{\Gamma, \tx}(e_1)
\concat \cdots \concat
\s{encEval}_{\Gamma, \tx}(e_\ell)
\bigr).
\]
\begin{align*}
\aeval(\inp_i a)
&\coloneqq
\begin{cases}
\s{aEval}_{\Gamma, \tx'}(a)
& \text{if $\s{prev}_i(\tx, \tx')$ for some $\tx' \in \Gamma$,} \\
\Omega
& \text{otherwise,}
\end{cases}
\\
\aeval(\oup_i a)
&\coloneqq
\begin{cases}
\s{aEval}_{\Gamma, \tx'}(a)
& \text{if $\s{succ}^{\Gamma}_i(\tx, \tx')$ for some $\tx' \in \Gamma$,} \\
\Omega
& \text{otherwise.}
\end{cases}
\end{align*}
\end{definition}

\begin{definition}[Group Expression Evaluation]
Evaluation $\geval \colon gexpr \to \mathbb{G} \cup \{\Omega\}$
in the context of ledger $\Gamma$ and transaction $\tx$ is defined by structural recursion, analogously to $\aeval$.
We again use strict propagation of $\Omega$:
if any required immediate subexpression evaluates to $\Omega$, then the whole expression evaluates to $\Omega$.
Otherwise:
\begin{align*}
\geval(\theta)
&\coloneqq \theta
&& (\theta \in \mathbb{G}),
\\
\geval(\s{var})
&\coloneqq \tx.\s{var}
&& (\s{var} \in S_g),
\\
\geval(\Omega)
&\coloneqq \Omega,
\\
\geval(g_1 \cdot g_2)
&\coloneqq \geval(g_1) \cdot \geval(g_2),
&
\geval(g^a)
&\coloneqq \geval(g)^{\aeval(a)}.
\end{align*}
\[
\geval\bigl(\s{H}(e_1, \ldots, e_\ell)\bigr)
\coloneqq
H_g\bigl(
\s{encEval}_{\Gamma, \tx}(e_1)
\concat \cdots \concat
\s{encEval}_{\Gamma, \tx}(e_\ell)
\bigr).
\]
\begin{align*}
\geval(\inp_i g)
&\coloneqq
\begin{cases}
\s{gEval}_{\Gamma, \tx'}(g)
& \text{if $\s{prev}_i(\tx, \tx')$ for some $\tx' \in \Gamma$,} \\
\Omega
& \text{otherwise,}
\end{cases}
\\
\geval(\oup_i g)
&\coloneqq
\begin{cases}
\s{gEval}_{\Gamma, \tx'}(g)
& \text{if $\s{succ}^{\Gamma}_i(\tx, \tx')$ for some $\tx' \in \Gamma$,} \\
\Omega
& \text{otherwise.}
\end{cases}
\end{align*}
\end{definition}

\begin{definition}[Boolean Expression Evaluation]
Let $\neg$, $\wedge$, and $\vee$ denote the connectives of Kleene's three-valued logic.
Evaluation $\beval \colon bexpr \to \{\bmbot, \s{U}, \bmtop\}$
in the context of ledger $\Gamma$ and transaction $\tx$ is defined by structural recursion.

The propositional clauses are:
\begin{align*}
\beval(\s{U})
&\coloneqq \s{U},
&
\beval(\bmbot)
&\coloneqq \bmbot,
\\
\beval(\neg \varphi)
&\coloneqq \neg\,\beval(\varphi),
\\
\beval(\varphi \wedge \psi)
&\coloneqq \beval(\varphi) \wedge \beval(\psi),
&
\beval(\varphi \vee \psi)
&\coloneqq \beval(\varphi) \vee \beval(\psi).
\end{align*}

The modal clauses are:
\begin{align*}
\beval(\inp_i \varphi)
&\coloneqq
\begin{cases}
\s{bEval}_{\Gamma, \tx'}(\varphi)
& \text{if } \exists \tx' \in \Gamma \colon \s{prev}_i(\tx, \tx'), \\
\s{U}
& \text{otherwise,}
\end{cases}
\\
\beval(\oup_i \varphi)
&\coloneqq
\begin{cases}
\s{bEval}_{\Gamma, \tx'}(\varphi)
& \text{if } \exists \tx' \in \Gamma \colon \s{succ}^{\Gamma}_i(\tx, \tx'), \\
\s{U}
& \text{otherwise.}
\end{cases}
\end{align*}

For the comparison and verification clauses, write
\[
A_i \coloneqq \aeval(a_i),
\qquad
G_i \coloneqq \geval(g_i)
\qquad (i \in \{1, 2\}).
\]
Then:
\begin{align*}
\beval(a_1 < a_2)
&\coloneqq
\begin{cases}
\s{U}
& \text{if } A_1 = \Omega \text{ or } A_2 = \Omega, \\
A_1 < A_2
& \text{otherwise,}
\end{cases}
\\
\beval(g_1 = g_2)
&\coloneqq
\begin{cases}
\s{U}
& \text{if } G_1 = \Omega \text{ or } G_2 = \Omega, \\
\bmtop
& \text{if } G_1 = G_2, \\
\bmbot
& \text{otherwise.}
\end{cases}
\end{align*}
\begin{align*}
\beval\bigl(\s{Ver}(g_1, g_2)\bigr)
&\coloneqq
\begin{cases}
\s{U}
& \text{if } G_1 = \Omega \text{ or } G_2 = \Omega, \\
\s{SigVer}'(G_1, M, G_2)
& \text{otherwise,}
\end{cases}
\\
\beval\bigl(\s{SVer}(g_1, a_1, g_2)\bigr)
&\coloneqq
\begin{cases}
\s{U}
& \text{if } G_1 = \Omega,\ A_1 = \Omega,\ \text{or } G_2 = \Omega, \\
\s{SigVer}(G_1, A_1, G_2)
& \text{otherwise.}
\end{cases}
\end{align*}

For RI propagation, we use the matcher
$\s{riMatch}_{\Gamma, \tx}$ defined in Definition~\ref{def:core_clauses}.
\[
\begin{aligned}
\beval\bigl(\inp_i \s{RI}(a_1, \ldots, a_m)\bigr)
&\coloneqq {}\\[-0.4ex]
&\quad
\begin{cases}
\s{riMatch}_{\Gamma, \tx}(\tx'; a_1, \ldots, a_m)
& \text{if } \exists \tx' \in \Gamma \colon \s{prev}_i(\tx, \tx'), \\
\s{U}
& \text{otherwise.}
\end{cases}
\end{aligned}
\]
\[
\begin{aligned}
\beval\bigl(\oup_i \s{RI}(a_1, \ldots, a_m)\bigr)
&\coloneqq {}\\[-0.4ex]
&\quad
\begin{cases}
\s{riMatch}_{\Gamma, \tx}(\tx'; a_1, \ldots, a_m)
& \text{if } \exists \tx' \in \Gamma \colon
	\s{succ}^{\Gamma}_i(\tx, \tx'), \\
\s{U}
& \text{otherwise.}
\end{cases}
\end{aligned}
\]
\end{definition}

Here, $\s{SigVer}'(\s{pk}, M, \sigma)$ verifies the signature $\sigma$ under public key $\s{pk}$ for the implicit message $M$.
The message $M$ is the hash of the immutable serialized transaction data, excluding signatures, ledger-derived metadata such as $\s{out}$ and $\oum$, and the canonical inclusion context.
Thus, realizing a later output does not change the message previously signed by the transaction.
$\s{SigVer}(\s{pk}, m, \sigma)$ verifies $\sigma$ under $\s{pk}$ for the explicit message $m$.

\section{Additional Case Studies and Evaluation Details}
\label{app:additional}

\subsection{Additional Case Studies}
\label{app:usecase}

\itpara{Gaming with Commit-and-Reveal.}
We sketch a three-player commit-and-reveal game where the winner is closest to a published aggregate.
Transaction $K_0$ (with $\lambda_1 = 0$) initiates the process and requires three commitment outputs.
Each commitment $U_i$ (with $\lambda_1 = 1$) can later be revealed by $V_i$ (with $\lambda_1 = 2$).
Transaction $K_1$ aggregates the three revealed values and determines the winner.
We avoid division by $3$ by comparing scaled distances to the sum.

In our DSL, the RI formula is:
\begin{align*}
	F \coloneqq &
	\left( \lambda_1 = 0 \to
		\left(
		\begin{array}{l}
			\oup_1 \s{RI}(1) \wedge \oup_2 \s{RI}(1)
				\wedge \oup_3 \s{RI}(1) \\
			\wedge\ \s{Ver}(y_1, \sigma_1) \ \wedge\
			\oup_4 \bmbot
		\end{array}
		\right)
	\right) \bigwedge \\
	&
	\left( \lambda_1 = 1 \to
		\left(
		\begin{array}{l}
			\oup_1 \s{RI}(2) \ \wedge \\
			\s{Ver}(y_1, \sigma_1) \ \wedge \\
			\oup_2 \bmbot
		\end{array}
		\right)
	\right) \bigwedge \\
	&
	\left( \lambda_1 = 2 \to
		\left(
		\begin{array}{l}
			\inp_2\bmbot \wedge \oup_2\bmbot \ \wedge \\
			\s{Ver}(y_1, \sigma_1) \ \wedge y_1 = \inp_1 y_1 \ \wedge \\
			\oup_1 \s{RI}(3) \wedge \inp_1 y_2 = gN^{x_1}
		\end{array}
		\right)
	\right) \bigwedge \\
	&
	\left( \lambda_1 = 3 \to
		\left(
		\begin{array}{l}
			x_1 = \inp_1 x_1 + \inp_2 x_1 + \inp_3 x_1 \\
			\wedge\ \s{Ver}(y_1, \sigma_1) \wedge y_1 = \oup_1 y_1 \\
			\wedge\ \inp_1 \s{RI}(3) \wedge \inp_2 \s{RI}(3) \wedge \inp_3 \s{RI}(3) \\
			\wedge\ \inp_4 \bmbot \wedge \oup_2 \bmbot \\
			\wedge\ \s{CorrWin}
		\end{array}
		\right)
	\right) ,
\end{align*}
where $\inp_i x_1$ is the revealed value from $V_i$, and $x_1$ in $K_1$ is their sum.
The winning condition is:
\begin{align*}
\s{CorrWin} & \coloneqq
\left(
	\left(
		|x_1 - 3 \cdot \inp_1 x_1| \le |x_1 - 3 \cdot \inp_2 x_1|
		\ \wedge\
		|x_1 - 3 \cdot \inp_1 x_1| \le |x_1 - 3 \cdot \inp_3 x_1|
	\right)
	\to \oup_1 y_1 = \inp_1 y_1
\right) \\
& \bigwedge
\left(
	\left(
		|x_1 - 3 \cdot \inp_2 x_1| < |x_1 - 3 \cdot \inp_1 x_1|
		\ \wedge\
		|x_1 - 3 \cdot \inp_2 x_1| \le |x_1 - 3 \cdot \inp_3 x_1|
	\right)
	\to \oup_1 y_1 = \inp_2 y_1
\right) \\
& \bigwedge
\left(
	\left(
		|x_1 - 3 \cdot \inp_3 x_1| < |x_1 - 3 \cdot \inp_1 x_1|
		\ \wedge\
		|x_1 - 3 \cdot \inp_3 x_1| < |x_1 - 3 \cdot \inp_2 x_1|
	\right)
	\to \oup_1 y_1 = \inp_3 y_1
\right).
\end{align*}

Here, $\oup_1 y_1$ designates the next spender, ensuring that the three token flows are spendable by the winner.

\itpara{Collective Receiving.}
Collective receiving is analogous to collective payments (\S\ref{subsec:usecase}), but aggregates incoming payers.
We use public keys $y_1$ to identify payers and an accumulator witness to authorize each contribution.
For three payers (each paying one coin), the transaction pattern is:
\begin{center}
~\xymatrixcolsep{4pc}\xymatrix{
\text{\mycircled{$K_0$}} \ar[r]^{\text{out} 1}_{\textsf{bal} = 0} &
\text{\mycircled{$K_1$}} \ar[r]^{\text{out} 1}_{\textsf{bal} = 1} &
\text{\mycircled{$K_2$}} \ar[r]^{\text{out} 1}_{\textsf{bal} = 2} &
\text{\mycircled{$K_3$}} \ar[r]^{\text{out} 1}_{\textsf{bal} = 3} &
U_r
\\
U_1 \ar@/_0.75pc/[ur]^{\text{out} 1}_{\textsf{bal} = 1} &
U_2 \ar@/_0.75pc/[ur]^{\text{out} 1}_{\textsf{bal} = 1} &
U_3 \ar@/_0.75pc/[ur]^{\text{out} 1}_{\textsf{bal} = 1} & &
\\
}
\end{center}

The corresponding RI formula is:
\begin{align*}
F \coloneqq
& ((0 \le \lambda_1) \wedge (\lambda_1 \le 3)) \ \wedge \\
& \Big(
	(\lambda_1 = 3 \wedge \oup_1 y_1 = \texttt{RECEIVER} \wedge \oup_2 \bmbot)
	\ \bigvee \\
& \quad
	(\lambda_1 < 3
	\wedge (\oup_1 y_2) = y_2^{\s{hashToInt}(y_1)} \wedge \s{Ver}(y_1, \sigma_1)
	\wedge \oup_1 \s{RI}(\lambda_1 + 1)
	\wedge (\oup_1 \inm_2) = 1 \wedge \oup_1 \inp_3 \bmbot)
\Big).
\end{align*}
The receiver spends transaction $K_3$ with $\lambda_1 = 3$.

\itpara{Rollback Micro-example (State cloning versus nesting).}
Consider a chain $\tx_0 \to \tx_1 \to \tx_2$ where $\tx_2$ must enforce a guarded rollback to state carried by $\tx_0$.
Without nested modality, formulas in $\tx_2$ can reference fields of $\tx_2$ and its direct inputs (fields of $\tx_1$), but not fields of inputs-of-inputs (fields of $\tx_0$).
A common encoding therefore clones state through $\tx_1$ so it remains available at $\tx_2$.
For example, if $\tx_0$ carries state variables $x_1, \ldots, x_\ell$, then $\tx_1$ carries fresh variables $x_{\ell + 1}, \ldots, x_{2\ell}$ with constraints $x_{\ell + k} = \inp_1 x_k$ for all $k \in [\ell]$.
Then $\tx_2$ enforces rollback by requiring $P \to \bigwedge_{k = 1}^{\ell} x_k = \inp_1 x_{\ell + k}$, where $P$ is a guard predicate.
With nested modality, the same dependence on $\tx_0$ can be stated directly as $P \to \bigwedge_{k = 1}^{\ell} x_k = \inp_1 \inp_1 x_k$, avoiding propagation solely for future checking.
\label{app:rollback}

\subsection{Baseline Instantiations and Workload Mapping}
\label{app:deploymentanchors}
\S\ref{sec:evaluation} records the external practice anchors for the selected workflow obligations and the scope of the baseline comparison.
This appendix documents the benchmark snapshot underlying Fig.~\ref{fig:evalRes};
the final RI specifications are given by the formal definitions and case-study formulas in the main paper.
Here we describe the workflow-level Solidity baselines used in the measurements.
These are reference baselines rather than deployed artifacts with fixed public ABIs, so we describe them by workflow slice, entry points, state roles, and stated exclusions.

\itpara{Evaluation Equivalence and Inputs.}
Our evaluation targets the on-chain enforcement layer and uses a workflow-level comparison rather than claiming globally optimal Solidity gas usage.
For each RI case study, we implement a minimal Solidity contract for the corresponding benchmark workflow slice (setup, per-step transition, and finalization when applicable) and its step-admission conditions.
Each benchmark transaction step or grouped multi-transaction segment
is matched to one or more EVM calls realizing the corresponding abstract workflow
transition and admission behavior.
We replay the workload trace and sum the gas of all EVM calls matched
to each unit.
For setup-heavy workflows, setup is reported separately;
for the layer-two asset case, we additionally report cumulative gas for the
grouped multi-transaction segments that emulate a fixed number of transfers.

All integer values are represented as $256$-bit words, and group elements as field elements modulo a prime $p < 2^{256}$.
The Solidity baselines use the corresponding $256$-bit calldata/storage representation, so the comparison does not rely on narrower machine-word encodings.
We model the on-chain commitment structure and transaction-shape constraints while treating off-chain dissemination of auxiliary validity data (\eg, witness packets or consignments) as outside the evaluation.
The evaluation measures workflow fidelity and per-step gas cost under the mapping above;
it does not provide full formal verification of the Solidity contracts or an exhaustive search over EVM designs or gas optimizations.
We treat the EVM baselines as transparent, functionally matched reference implementations for the selected workloads.
The Solidity counterparts used in the measurements are available in the code repository\footnote{\url{https://github.com/htftsy/CTVPlus}.} and are included in the artifact together with the benchmark RI encodings and recorded outputs used in the reported comparison.

\itpara{Non-Fungible Tokens.}
For NFTs, we do not model the full ERC-721 engineering stack.
Instead, we extract the ownership-transition slice evaluated here:
minting creates a token identifier with an owner, and transfer preserves the identifier while updating the owner under authorization~\cite{EIP721,OZContracts}.
The RI encoding realizes this obligation as a per-token ownership-transition workflow carried by the designated UTXO state.
The Solidity baseline realizes this obligation through workflow-level entry points \texttt{mint} and \texttt{transfer}, with state and call arguments tracking the token identifier, the current owner, and the successor holder.
Metadata, enumeration, approvals, receiver hooks, and marketplace logic
are outside the evaluated ownership-transition slice.

\itpara{Gaming with Guesses and an External Beacon.}
The gaming case models setup, guessing, and finalization.
The RI uses $\lambda_1 \in \{0, 1, 2\}$ to encode these three stages.
The Solidity baseline realizes setup in the constructor, records guesses with \texttt{guess}, and finalizes with \texttt{determineWinner}.
Variable $x_1$ is the guessed value, while $y_1$ and $\sigma_1$ capture step authorization in the RI.
The value $\s{Rnd}$ corresponds to block-derived or oracle-supplied randomness used only at finalization~\cite{SolidityGlobals,EIP4399,ChainlinkVRF}.
The baseline stores the two guesses and player addresses, while the contract balance holds the settlement stakes.
Game-platform functionality beyond this staged guess-and-settle workflow
is outside scope.

\itpara{Layer-Two Assets with UTXO Anchoring.}
The case study models split, transfer, and merge over UTXO-anchored assets.
The Solidity baseline uses a fungible-balance contract whose constructor initializes the asset supply and whose \texttt{transfer} entry point updates balances.
The RI split/merge workflow is mapped to ordinary transfers realizing the corresponding net asset movements, as detailed below.
Matching is by asset conservation and net ownership/amount transitions, not by explicit split/merge transaction structure in Solidity.
The surrounding RGB- and Taproot-Assets-style engineering stack is outside scope~\cite{RGBClientSideValidation24,RGBSingleUseSeals24,RGBStateTransitions25,LightningTaprootAssetsProtocol24,LightningTaprootAssetsOverview24}.

\itpara{Collective Payments.}
We model a sender-funded sequential payout in which authorized recipients claim fixed amounts using accumulator witnesses.
The RI carries the payout state forward across claims and enforces the corresponding witness and transaction-shape conditions.
The Solidity baseline receives the funds through \texttt{charge} and uses repeated \texttt{recvMoney} calls to verify recipients and release payments.
The matched slice is the sequential payout itself;
wallet coordination, network negotiation, privacy heuristics, and coin-selection policy are outside scope.

\itpara{Taxed Transfers / Transfer-tax Tokens.}
The benchmark models deterministic fee routing on each transfer.
The Solidity baseline stores the controlling account, the tax recipient, and the aggregate funded amount.
Its \texttt{sendMoney} entry point takes the recipient, transfer amount, and tax amount, enforces the prescribed ratio, and routes both amounts accordingly.
The matched obligation is deterministic transfer-tax routing;
router integration and exchange-compatibility machinery are outside scope~\cite{UniswapRouter02,UniswapCommonErrors,UniswapFoTIssue}.

\itpara{Rate-Limited Token Flow.}
The benchmark models a cooldown policy that enforces a minimum inter-spend gap until an expiration point.
The Solidity baseline uses a workflow-level spending entry point with state recording the previous accepted step, the configured minimum gap, and the expiration point.
Before expiration, each spend must satisfy the required gap from the previous accepted spend and carry the policy forward.
The first policy-bearing spend at or after expiration still satisfies the gap check but terminates the policy;
subsequent spends are unrestricted by it.
The matched obligation is the cooldown/expiry admission rule;
the surrounding wallet-policy framework is outside scope~\cite{SafeSpendingLimits,SafeModules}.

\itpara{Step-by-Step Workload Mapping.}
The exact mapping used in the measurements is as follows.
\begin{enumerate}[label=(\alph*)]
\item \emph{NFT.}
$K_0$ is simulated by the constructor.
$K_1$ and $K_2$ are simulated by \texttt{mint}.
Each subsequent transaction corresponds to \texttt{transfer}.

\item \emph{Gaming with Guesses and an External Beacon.}
The initial transaction $K_0$ is simulated by the contract constructor.
$U_0$ and $U_1$ are simulated by \texttt{guess}, which inputs a guess $x_1$.
$K_1$ is simulated by \texttt{determineWinner}.

\item \emph{Layer-Two Assets.}
$U_1$ is simulated by the constructor.
Our sequence consisting of transfer ($V_1$), split ($V_2$--$V_3$), transfer witnesses ($U_2, U_3$), and merge ($V_4$) emulates four EVM layer-two asset transfers.
These correspond to a generator-to-user transfer ($V_1$), a user-to-user transfer of the first split part ($V_2, U_2$), a user-to-user transfer of the second split part ($V_3, U_3$), and a merge that recombines the two parts for the third recipient ($V_4$).
$V_1$, $V_4$, $V_2 + U_2$ (the net effect of $V_2$ and $U_2$), and $V_3 + U_3$ are each simulated by \texttt{transfer}.

\item \emph{Collective Payments.}
$K_0$ is simulated by the constructor, the charging function \texttt{charge}, and one invocation of \texttt{recvMoney}, as the proposer of $K_0$ is also one receiver to simplify the RI.
Each subsequent $K_i$ is simulated by \texttt{recvMoney}.

\item \emph{Taxed Transfers.}
$K_0$ is simulated by the constructor, \texttt{charge}, and \texttt{sendMoney}.
Each subsequent $K_i$ is simulated by \texttt{sendMoney}.

\item \emph{Rate-Limited Token Flow.}
$K_1$ is simulated by the constructor.
$K_i$ for $i \ge 2$ is simulated by \texttt{send}.
\end{enumerate}
We keep each Solidity contract minimal and aligned with the workflow semantics of the corresponding RI encoding.
When applicable, the constructor simulates the setup transaction, and each later step is simulated by a dedicated function call (\eg, \texttt{transfer}, \texttt{send}, or \texttt{recvMoney}).

\itpara{Public Implementation Pointers.}
For additional implementation context, we list public repositories with functionality related to several use cases.
These are illustrative pointers, not the practice anchors in Table~\ref{tab:anchors} or the EVM baselines used in our evaluation:
\begin{itemize}
\item Gaming with guesses and an external beacon:
\url{https://github.com/maAPPsDEV/coin-flip}.
\item Token ownership and transfer:
\url{https://github.com/nibbstack/erc721} and
\url{https://github.com/softstack/erc20-token-template}.
\item Collective payments:
\url{https://github.com/divyalalwani/Batch-Contract} and
\url{https://github.com/0xcaff/splitter-contract}.
\item Taxed transfers:
\url{https://github.com/Giulio2002/Ethereum-Taxable-Token}.
\end{itemize}

\section{Limitations and Future Directions}
\label{app:future}

\itpara{Limitations.}
Beyond the scope boundaries in \S\ref{subsec:limitations}, we highlight three implementation-level limitations of the checker/DSL design.
First, general-purpose references, arrays, and maps are unsupported, so
large structured state must be encoded directly (e.g., through graph encodings),
which can be expensive.
Second, the expression language has no intra-transaction loop construct;
iterative behavior is encoded through transaction-chain workflows.
Third, the checker design supports at most one RI formula and one RI instance
per transaction, limiting some composability patterns.

\itpara{Conjunction of RIs.}
Currently, each transaction carries at most one RI formula and one RI-instance label;
the instance is either fresh or inherited from the unique label induced by its input(s).
Complex workflows may require combining multiple RI obligations, for example, a transaction subject to both taxation and temporal constraints.
This motivates mechanisms for safe conjunction of invariants.
One possible extension is to allow outputs to carry an optional $(\mathit{ri\_id}, \lambda)$ pair and to invoke the RI interpreter on each attached formula over the current transaction, its inputs, and its outputs.
RI-enabled outputs could then be treated as another script flavor, allowing deployment through the usual transaction-layer extension pattern.

\itpara{Production-quality DSLs.}
The current DSL and benchmark tooling operate close to the logical RI
representation.
A production-quality developer-facing layer would require higher-level
abstractions, stronger typing, modular composition, and actionable diagnostics
while preserving validation-time evaluability and integration with existing
build and verification workflows.
Production-grade compilation and execution back ends for Plutus-, Move-,
Solidity-style toolchains and other virtual machines are outside the current tooling.

\itpara{Outer Logic for Ledger.}
An outer meta-logic could reason over ledger extensions and transaction patterns using per-transaction RIs and the inner logic of \S\ref{sec:dslri}.
Such reasoning requires trace-level assumptions beyond the semantics developed here.
The inner semantics exposes designated outputs and one-step re-checking as an interface for this outer layer.
Two natural extensions are a trace-based logic that proves when $\s{U}$ resolves at designated spends and a compositional treatment that lifts
per-obligation progress to whole workflows.

\itpara{Black-hole RIs and Satisfiability.}
Stronger tool and language support could reason about RI satisfiability.
Black-hole behavior is not unique to RIs, but RI conditions are logical
formulas drawn from a fixed, intentionally small constraint language,
providing a structured basis for static checks, counterexample generation,
and developer guidance.
Satisfiability-oriented tooling is outside the current tooling.

\section{Proof of Soundness}
\label{app:detailproof}

\begin{theorem}[Soundness]
For any valid ledger $\Gamma$, transaction $\tx$, and proposition $p \in bexpr$,
if $\Gamma, \tx \vdash p$, then $\Gamma, \tx \models p$.
\end{theorem}

\begin{proof}
Fix a valid ledger $\Gamma$ and a transaction $\tx$.
We prove the theorem by induction on the last rule used in a derivation of $\Gamma, \tx \vdash p$.
We use the following facts repeatedly.
First, $\Gamma, \tx \models q$ iff $\tx \neq \bmbot_{\tx}$ and $\beval(q) \in \nonf$.
Second, for $\rhd \in \{\inp_i, \oup_i\}$, if the referenced transaction does not exist, then $\beval(\rhd q) = \s{U}$ by definition of $\s{bEval}$.
Third, all propositional identities used below are checked by the truth tables of Kleene's three-valued logic.

\begin{itemize}[topsep = 0pt]

\item\emph{$\wedge$-introduction.}
Assume $\Gamma, \tx \vdash \varphi$ and $\Gamma, \tx \vdash \psi$.
By the induction hypotheses, $\Gamma, \tx \models \varphi$ and $\Gamma, \tx \models \psi$.
Hence $\beval(\varphi), \beval(\psi) \in \nonf$.
By the K3 truth table for $\wedge$, this implies
$\beval(\varphi \wedge \psi) = \beval(\varphi) \wedge \beval(\psi) \in \nonf$.
Therefore, $\Gamma, \tx \models \varphi \wedge \psi$.

\item\emph{$\to$-introduction.}
Assume $\Gamma, \tx \vdash \neg \varphi \vee \psi$.
By the induction hypothesis, $\Gamma, \tx \models \neg \varphi \vee \psi$.
By syntax, $\varphi \to \psi$ is an abbreviation for $\neg \varphi \vee \psi$.
Hence $\beval(\varphi \to \psi) = \beval(\neg \varphi \vee \psi) \in \nonf$.
Therefore, $\Gamma, \tx \models \varphi \to \psi$.

\item\emph{$\wedge$-elimination left/right.}
Assume $\Gamma, \tx \vdash \varphi \wedge \psi$.
By the induction hypothesis, $\Gamma, \tx \models \varphi \wedge \psi$.
Hence $\beval(\varphi \wedge \psi) = \beval(\varphi) \wedge \beval(\psi) \in \nonf$.
By the K3 truth table for $\wedge$, this implies $\beval(\varphi) \in \nonf$ and $\beval(\psi) \in \nonf$.
Therefore, $\Gamma, \tx \models \varphi$ and $\Gamma, \tx \models \psi$.

\item\emph{Exclusion.}
We must show $\Gamma, \tx \models \varphi \vee \neg \varphi$.
Let $v \coloneqq \beval(\varphi)$.
If $v \in \{\bmbot, \bmtop\}$, then $v \vee \neg v = \bmtop$.
If $v = \s{U}$, then $v \vee \neg v = \s{U}$.
In all cases, $\beval(\varphi \vee \neg \varphi) \in \nonf$.
Therefore, $\Gamma, \tx \models \varphi \vee \neg \varphi$.

\item\emph{ex falso quodlibet (without $\rhd$).}
Assume $\Gamma, \tx \vdash \bmbot$.
By the induction hypothesis, $\Gamma, \tx \models \bmbot$.
But $\beval(\bmbot) = \bmbot \notin \nonf$, contradiction.
Hence this premise is semantically impossible.
Therefore, the conclusion $\Gamma, \tx \models \varphi$ follows vacuously.

\item\emph{ex falso quodlibet (with $\rhd \in \{\inp_i, \oup_i\}$).}
Assume $\Gamma, \tx \vdash \rhd \bmbot$.
By the induction hypothesis, $\Gamma, \tx \models \rhd \bmbot$.
We show $\Gamma, \tx \models \rhd \varphi$.

If the referenced predecessor/successor transaction does not exist, then
$\beval(\rhd \bmbot) = \s{U}$ and $\beval(\rhd \varphi) = \s{U}$ by definition of $\s{bEval}$.
Hence $\Gamma, \tx \models \rhd \varphi$.

Otherwise, let $\tx' \in \Gamma$ be the referenced transaction.
Then $\beval(\rhd \bmbot) = \s{bEval}_{\Gamma, \tx'}(\bmbot) = \bmbot \notin \nonf$,
contradicting $\Gamma, \tx \models \rhd \bmbot$.
Thus this branch is impossible.

Therefore, in all cases, $\Gamma, \tx \models \rhd \varphi$.

\item\emph{$\vee$-introduction left/right.}
We prove the left case.
Assume $\Gamma, \tx \vdash \varphi$.
By the induction hypothesis, $\Gamma, \tx \models \varphi$, so $\beval(\varphi) \in \nonf$.
By the K3 truth table for $\vee$, $\beval(\varphi \vee \psi) = \beval(\varphi) \vee \beval(\psi) \in \nonf$.
Hence $\Gamma, \tx \models \varphi \vee \psi$.
The right case is symmetric.

\item\emph{$\vee$-introduction (the implication-combining rule).}
Assume $\Gamma, \tx \vdash \varphi \to \psi$ and $\Gamma, \tx \vdash \varphi' \to \psi$.
By the induction hypotheses,
$\beval(\varphi \to \psi) \in \nonf$
and
$\beval(\varphi' \to \psi) \in \nonf$.
We must show $\beval((\varphi \vee \varphi') \to \psi) \in \nonf$.

Write
$x \coloneqq \beval(\varphi)$,
$y \coloneqq \beval(\varphi')$,
and
$z \coloneqq \beval(\psi)$.
Then
$\beval(\varphi \to \psi) = \neg x \vee z \in \nonf$
and
$\beval(\varphi' \to \psi) = \neg y \vee z \in \nonf$.

If $z \in \nonf$, then
$\beval((\varphi \vee \varphi') \to \psi) = \neg(x \vee y) \vee z \in \nonf$
immediately.

If $z = \bmbot$, then $\neg x \vee z \in \nonf$ implies $x \in \{\bmbot, \s{U}\}$,
and similarly $y \in \{\bmbot, \s{U}\}$.
Hence $x \vee y \in \{\bmbot, \s{U}\}$, so $\neg(x \vee y) \in \nonf$.
Therefore, $\neg(x \vee y) \vee z \in \nonf$.

In both cases, $\Gamma, \tx \models (\varphi \vee \varphi') \to \psi$.

\item\emph{De Morgan.}
Assume
$\Gamma, \tx \vdash \neg(\varphi_1 \wedge \varphi_2 \wedge \ldots \wedge \varphi_\ell)$.
By the induction hypothesis,
$\Gamma, \tx \models \neg(\varphi_1 \wedge \cdots \wedge \varphi_\ell)$.
We must show
$\Gamma, \tx \models \neg\varphi_1 \vee \neg\varphi_2 \vee \cdots \vee \neg\varphi_\ell$.

This is a propositional K3 identity.
A direct induction on $\ell$ suffices.
The base case $\ell = 1$ is immediate.
The step from $\ell - 1$ to $\ell$ follows by applying the binary De Morgan law to
$\neg\bigl((\varphi_1 \wedge \cdots \wedge \varphi_{\ell - 1}) \wedge \varphi_\ell\bigr)$
and then using the induction hypothesis on the prefix.
Hence the conclusion holds.

\item\emph{Distribution.}
Assume $\Gamma, \tx \vdash (\varphi_1 \vee \varphi_2) \wedge \psi$.
By the induction hypothesis, $\Gamma, \tx \models (\varphi_1 \vee \varphi_2) \wedge \psi$.
We must show $\Gamma, \tx \models (\varphi_1 \wedge \psi) \vee (\varphi_2 \wedge \psi)$.

This is again a propositional K3 identity.
Let
\[
x \coloneqq \beval(\varphi_1), \qquad
y \coloneqq \beval(\varphi_2), \qquad
z \coloneqq \beval(\psi).
\]
Then $(x \vee y) \wedge z \in \nonf$ implies
$(x \wedge z) \vee (y \wedge z) \in \nonf$ by the K3 truth table.
Therefore,
$\Gamma, \tx \models (\varphi_1 \wedge \psi) \vee (\varphi_2 \wedge \psi)$.

\item\emph{Disjunctive distribution ($\vee$D).}
Assume $\Gamma, \tx \vdash \rhd(\varphi \vee \psi)$.
By the induction hypothesis, $\Gamma, \tx \models \rhd(\varphi \vee \psi)$.
We show $\Gamma, \tx \models \rhd\varphi \vee \rhd\psi$.

If the referenced predecessor/successor transaction does not exist, then
$\beval(\rhd(\varphi \vee \psi)) = \beval(\rhd\varphi) = \beval(\rhd\psi) = \s{U}$.
Hence $\beval(\rhd\varphi \vee \rhd\psi) = \s{U} \in \nonf$.

Otherwise, let $\tx'$ be the referenced transaction.
Then
\begin{align*}
\beval(\rhd(\varphi \vee \psi))
&= \s{bEval}_{\Gamma, \tx'}(\varphi \vee \psi)
= \s{bEval}_{\Gamma, \tx'}(\varphi) \vee \s{bEval}_{\Gamma, \tx'}(\psi) \\
&= \beval(\rhd\varphi) \vee \beval(\rhd\psi) \\
&= \beval(\rhd\varphi \vee \rhd\psi).
\end{align*}
Thus, $\beval(\rhd\varphi \vee \rhd\psi) \in \nonf$, and $\Gamma, \tx \models \rhd\varphi \vee \rhd\psi$.

\item\emph{K-distribution (K).}
Assume $\Gamma, \tx \vdash \rhd(\varphi \to \psi)$.
By the induction hypothesis, $\Gamma, \tx \models \rhd(\varphi \to \psi)$.
We show $\Gamma, \tx \models \rhd\varphi \to \rhd\psi$.

If the referenced predecessor/successor transaction does not exist, then
$\beval(\rhd(\varphi \to \psi)) = \beval(\rhd\varphi) = \beval(\rhd\psi) = \s{U}$,
so
$\beval(\rhd\varphi \to \rhd\psi) = \neg \s{U} \vee \s{U} = \s{U} \in \nonf$.

Otherwise, let $\tx'$ be the referenced transaction.
Then
\begin{align*}
\beval(\rhd(\varphi \to \psi))
&= \s{bEval}_{\Gamma, \tx'}(\varphi \to \psi)
= \neg \s{bEval}_{\Gamma, \tx'}(\varphi) \vee \s{bEval}_{\Gamma, \tx'}(\psi) \\
&= \neg \beval(\rhd\varphi) \vee \beval(\rhd\psi) \\
&= \beval(\rhd\varphi \to \rhd\psi).
\end{align*}
Thus, $\Gamma, \tx \models \rhd\varphi \to \rhd\psi$.

\item\emph{Conjunctive distribution ($\wedge$D).}
Assume $\Gamma, \tx \vdash \rhd(\varphi \wedge \psi)$.
By the induction hypothesis, $\Gamma, \tx \models \rhd(\varphi \wedge \psi)$.
We show $\Gamma, \tx \models \rhd\varphi \wedge \rhd\psi$.

If the referenced predecessor/successor transaction does not exist, then
$\beval(\rhd(\varphi \wedge \psi)) = \beval(\rhd\varphi) = \beval(\rhd\psi) = \s{U}$,
so
$\beval(\rhd\varphi \wedge \rhd\psi) = \s{U} \wedge \s{U} = \s{U} \in \nonf$.

Otherwise, let $\tx'$ be the referenced transaction.
Then
\begin{align*}
\beval(\rhd(\varphi \wedge \psi))
&= \s{bEval}_{\Gamma, \tx'}(\varphi \wedge \psi)
= \s{bEval}_{\Gamma, \tx'}(\varphi) \wedge \s{bEval}_{\Gamma, \tx'}(\psi) \\
&= \beval(\rhd\varphi) \wedge \beval(\rhd\psi) \\
&= \beval(\rhd\varphi \wedge \rhd\psi).
\end{align*}
Thus, $\Gamma, \tx \models \rhd\varphi \wedge \rhd\psi$.

\item\emph{Tautology (T).}
Assume $\Gamma, \tx \vdash \varphi_1 \to \varphi_2 \to \psi$.
By the induction hypothesis, $\Gamma, \tx \models \varphi_1 \to \varphi_2 \to \psi$.
We show $\Gamma, \tx \models \varphi_1 \wedge \varphi_2 \to \psi$.

Write
$x \coloneqq \beval(\varphi_1)$,
$y \coloneqq \beval(\varphi_2)$,
and
$z \coloneqq \beval(\psi)$.
Then
\[
\beval(\varphi_1 \to \varphi_2 \to \psi)
= \neg x \vee \neg y \vee z
\in \nonf.
\]
We need to show
\[
\beval((\varphi_1 \wedge \varphi_2) \to \psi)
= \neg(x \wedge y) \vee z
\in \nonf.
\]
This is a propositional K3 identity
(equivalently, a De Morgan instance plus associativity of~$\vee$).
Hence the conclusion follows.

\end{itemize}
\end{proof}

This implies the corollary below.

\begin{corollary}
For any valid ledger $\Gamma$ and proposition $p \in bexpr$, if $\Gamma \vdash p$, then $\Gamma \models p$.
\end{corollary}

\begin{proof}
By definition, $\Gamma \vdash p$ implies $\Gamma, \tx \vdash p$ for each $\tx\in\Gamma$.
By soundness, $\Gamma, \tx \models p$ holds for each $\tx\in\Gamma$.
Hence $\Gamma \models p$.
\end{proof}

\fi

\end{document}